\PassOptionsToPackage{dvipsnames}{xcolor}
\documentclass[a4paper,twocolumn,11pt,unpublished, allowtoday]{quantumarticle}
\pdfoutput=1
\usepackage[utf8]{inputenc}
\usepackage[english]{babel}
\usepackage[T1]{fontenc}
\usepackage{graphicx}
\usepackage{dcolumn}
\usepackage{bm}
\usepackage{braket}
\usepackage{amsthm}
\usepackage{amsmath,amssymb}
\usepackage{mathtools}
\usepackage[numbers,sort&compress]{natbib}
\usepackage{hyperref}

\newtheorem*{theorem*}{Theorem} % 별표 버전 정의
\newtheorem{fact}{Fact}
\newtheorem{theorem}{Theorem}

\theoremstyle{definition}
\newtheorem{definition}{Definition}
\theoremstyle{remark}
\newtheorem{remark}{Remark}
\allowdisplaybreaks

\providecommand{\onecolumngrid}{\onecolumn}

\begin{document}

%\preprint{APS/123-QED}

\title{Sample-optimal single-copy quantum state tomography with shallow-depth measurements
}% 

\author{Gyungmin Cho}
 \email{km950501@snu.ac.kr}
 \affiliation{NextQuantum Center, Department of Physics and Astronomy, and Institute of Applied Physics, Seoul National University, Seoul 08826, Korea}%Lines break automatically or can be forced with \\
\author{Dohun Kim}%
 \email{dohunkim@snu.ac.kr}
 \affiliation{NextQuantum Center, Department of Physics and Astronomy, and Institute of Applied Physics, Seoul National University, Seoul 08826, Korea}%


\begin{abstract}
    Quantum state tomography (QST) is a central task in quantum information, and its efficiency is commonly characterized by sample complexity. Although collective measurements on multiple copies achieve optimal performance, they are difficult to implement on near-term devices, motivating the study of single-copy approaches. Here, we introduce an ancilla-free single-copy QST protocol based on logarithmic-depth local circuits on an $n$-qubit system. For rank-$r$ states in dimension $d=2^n$, our protocol achieves trace-norm error $\epsilon$ using $\mathcal{O}(dr^2\log d/\epsilon^2)$ copies, matching the single-copy lower bound up to a logarithmic factor. For full-rank mixed states, it removes this logarithmic overhead and achieves the optimal scaling $\mathcal{O}(d^3/\epsilon^2)$, with nearly optimal classical runtime for explicit matrix output. These results show that sample-optimal QST can be realized using experimentally accessible shallow-depth measurements.
\end{abstract}

\maketitle

%\tableofcontents
\section{Introduction}
In quantum mechanics, the state space of an $n$-qubit system is given by the tensor product of the Hilbert spaces of the individual qubits, in accordance with the postulate for composite systems \cite{Nielsen_Chuang_2010}. 
Thus, an $n$-qubit quantum state is represented by a positive semidefinite
matrix $\rho \in \mathbb{C}^{d\times d}$ with unit trace, where $d = 2^n$, and admits
the spectral decomposition $\rho = \sum_{i=1}^{r} \lambda_i \ket{\psi_i}\bra{\psi_i}$,
with rank $r \le d$. If $r=1$, the state is called a \emph{pure state}, otherwise, it is a \emph{mixed state}. 

Although a quantum state is mathematically represented by a density matrix, the density matrix itself is not directly accessible in experiments. 
Instead, information about the state is inferred from measurement outcomes. 
In the measurement models considered here, this information is obtained by preparing independent copies of the state and measuring them under different measurement settings, with repeated measurements used to control statistical fluctuations. 
The task of reconstructing the density matrix from such measurement data is called quantum state tomography (QST).
In this context, the number of prepared copies of the quantum state $\rho$, referred to as the \emph{sample complexity}, serves as a standard figure of merit for QST.

Since the number of parameters required to describe a state $\rho$ grows exponentially with $n$, the sample complexity required for QST also scales exponentially.
As a result, fully reconstructing the entire quantum state becomes infeasible for large $n$. Nevertheless, in the early stages of quantum computer development, QST has been widely employed for device benchmarking \cite{qst_ion, qst_nv, qst_super}. 

In some cases, however, it is sufficient to perform QST not on the entire $n$-qubit system 
but only on a specific subsystem. For example, in $t$-doped stabilizer state learning, once the 
stabilizer part is learned, full QST is required only on a certain subsystem \cite{hangleiter2024bell, CGYZ25, grewal2024improved, leone2024learning}. Likewise, in tasks where one assumes access to a target quantum state and seeks to approximate the circuit that prepares it \cite{landau2025learning}, subsystem QST is used as a subroutine. Therefore, while the exponential sample complexity required for QST cannot be avoided, achieving the optimal sample complexity remains highly beneficial for a wide range of applications.

The accuracy of QST can be quantified using different metrics, among which fidelity and trace distance are most common \cite{haah2016sample, chen2023does}.
Here, we adopt the trace distance and formalize the QST problem as follows:

\begin{definition} $(\epsilon, \delta)$-QST.
Given $\epsilon > 0$ and $\delta \in (0,1)$, let $\rho$ denote the target quantum state and 
$\hat{\rho}_T$ its estimator. 
The goal is to determine a bound on the number of samples $T$ such that 
\[
\Pr\!\left( \|\hat{\rho}_T - \rho\|_{\mathrm{tr}} \geq \epsilon \right) \leq \delta .
\]
\end{definition}

For QST of a rank-$r$ quantum state, it is known that multi-copy measurements achieve the optimal sample complexity \cite{haah2016sample, o2016efficient} given by
\[
T = \Theta\big(dr / \epsilon^2 \big).
\]
However, current quantum computers are still noisy and lack error correction, which limits the number of copies that can be used simultaneously. 
Consequently, implementing such an optimal QST strategy is highly challenging.  

Follow-up research~\cite{chen2024optimal} has investigated scenarios in which there is a restriction on 
the number of copies $t$ that can be prepared and measured simultaneously. 
It has been shown that in such cases one needs $T=\widetilde{\Theta}(d^3/\sqrt{t}\epsilon^2)$ in the full rank case, 
meaning that the benefit is only achieved once $t$ becomes sufficiently large. 
Considering that even implementing 2-copy measurements ($t=2$) is difficult on present-day quantum hardware, the practical limitations of QST using multi-copy measurements become evident.  

Compared to the multi-copy measurements, a more practical approach 
is to employ single-copy measurements. 
In this case, the optimal sample complexity of QST is known to be $\Theta(dr^2/\epsilon^2)$. However, implementing this strategy requires somewhat complicated positive-operator valued measurements (POVMs) \cite{haah2016sample, kueng2017low} or projections 
onto random Haar states \cite{guctua2020fast}. By performing projective measurements drawn from a state 2-design \cite{Lowe2025lower}, a near-optimal strategy---incurring an additional $\ln(d)$ factor---can be obtained.
Nevertheless, even this scheme demands circuits of depth ${\cal O}(n)$. Considering device connectivity and noise, such requirements remain impractical on current noisy quantum hardware.

\subsection{Contributions}
Here, motivated by recent advances \cite{schuster2025random} in constructing approximate unitary designs with shallow-depth circuits, we investigate whether shallow circuits can attain the optimal sample complexity for single-copy measurements [Fig.~\ref{fig1}]. A detailed discussion of related work is provided in Appendix~\ref{appx:related_works}.

\begin{figure}[t]
    \centering
    \includegraphics[width=1\linewidth, trim=2.5cm 5.7cm 4.5cm 5.5cm, clip]{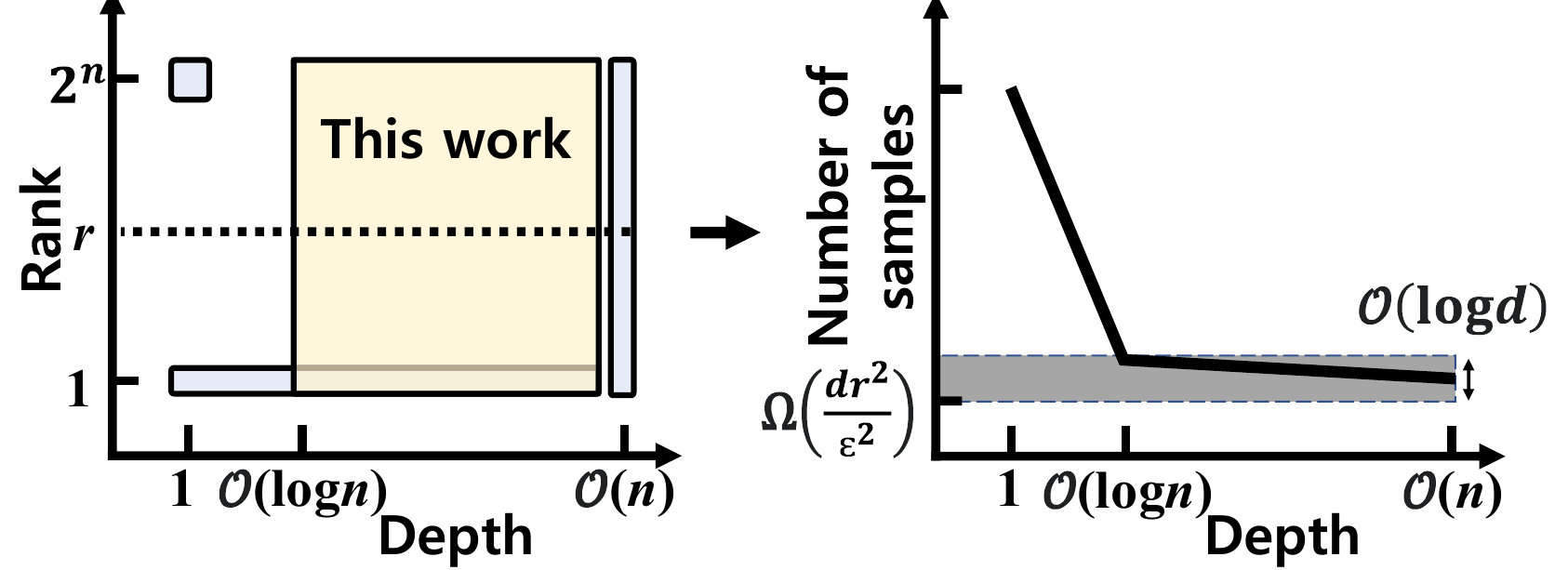}
    \caption{(Left) Previously, near-optimal quantum state tomography (QST) was studied using either the $n$-qubit Clifford group $\mathrm{Cl}(n)$ or Pauli measurements (blue). Here, we show that shallow-depth circuits extend the near-optimal region (yellow), and that in the full-rank case $r=d$ they become optimal.
    (Right) We compute an upper bound on the number of samples required for QST, and show that with circuit depth $\mathcal{O}(\log n)$, this bound is close to the single-copy lower bound $\Omega(dr^2/\epsilon^2)$, up to a $\ln d$ factor.}
    \label{fig1}
\end{figure}

Previous studies have shown that, by using the biased estimator $\hat{\rho}_{\text{biased}}={\cal M}^{-1}_{\text{Haar}}(U^{\dagger}\ket{b}\!\bra{b}U)=(d+1)U^{\dagger}\ket{b}\!\bra{b}U-I$,
one can estimate the fidelity $\mathrm{Tr}(\ket{\psi}\!\bra{\psi}\rho)$ to within an additive error $\epsilon$ using only ${\cal O}(1/\epsilon^2)$ samples \cite{schuster2025random}.
However, applying this approach to QST poses challenges. 
Specifically, if QST guarantees a small trace norm 
$\|\rho - \hat{\rho}\|_\text{tr} \leq \epsilon$, 
then for any operator $O$ one must satisfy
$\big|\mathrm{Tr}(O\rho) - \mathrm{Tr}(O\hat{\rho})\big| 
    \leq \|O\|_{\text{op}} \, \epsilon $.
In contrast, there are simple counterexamples showing that with the biased estimator this guarantee cannot be maintained at depth ${\cal O}(\log n)$, and therefore a different estimator is required for QST [see Appendix~\ref{appx:biased}].

Constructing an unbiased estimator for $\rho$ is not technically difficult in itself, 
but proving the concentration of its empirical average is nontrivial. 
In this work, instead of relying on approximate unitary designs, we leverage 
the same circuit ansatz used in prior studies \cite{schuster2025random, cho2025entanglement} and provide a proof of concentration 
for the unbiased estimator. 
Technically, our proof combines the Weingarten calculus \cite{collins2022weingarten} with the transfer matrix method  \cite{transfer}. This method may also be useful for various analyses employing shallow randomized measurements. 

\section{Preliminaries}
We now define the variables and basic concepts required for the subsequent discussion. 

First, to construct an unbiased estimator of $\rho$, we introduce the shadow channel ${\cal M}$~\cite{huang2020predicting}, defined by
\begin{equation}
    \mathcal{M}(\rho) = \mathbb{E}_{U\sim{\cal E},b}U^{\dagger}\ket{b}\!\bra{b}U,
\end{equation}
where $\mathcal{E}$ is the unitary ensemble from which $U$ is uniformly sampled and $b\in{\{0,1\}}^n$ is the bitstring of the measurement outcome. 
Here, we use a unitary ensemble consisting only of Clifford gates \cite{cho2025entanglement}. 
Then, each Pauli operator $P\in\{I,X,Y,Z\}^n$ is an eigenoperator of the shadow channel $\mathcal{M}(P) = m_P P$,
where $m_P = \mathbb{E}_{U \sim \mathcal{E}}\!\left[\mathbf{1}\{UPU^{\dagger}\in\pm{\cal Z}\}\right]$, with $\mathbf{1}\{\text{True}\}=1$ (and $0$ otherwise), and ${\cal Z}=\{I,Z\}^n$.
Using the shadow channel $\mathcal{M}$, the unbiased estimator of the quantum state $\rho$ is given by
\begin{equation}\label{eq:estimator_QST1}
    \hat{\rho} = \mathcal{M}^{-1}\big(U^{\dagger} \ket{b}\!\bra{b} U\big).
\end{equation}
If the unitary ensemble $\mathcal E$ is informationally complete, then ${\cal M}^{-1}$ exists and $\mathcal{M}^{-1}(P) = m_P^{-1} P$. 
For a given unitary ensemble $\mathcal{E}$, we define the Pauli correlation function $\tau(P, Q)$~\cite{bertoni2024shallow} as $\tau(P, Q; n) = \mathbb{E}_{U\sim {\cal E}}\!\left[\mathbf{1}\{UPU^{\dagger}\in\pm{\cal Z}\}\mathbf{1}\{UQU^{\dagger}\in\pm{\cal Z}\}\right]$,
where $n$ denotes the system size on which Pauli operators $P$ and $Q$ are defined. We omit $n$ when clear from context.

For $\mathcal{E} = \mathrm{Cl}(k)^{\otimes n/k}$, $\tau(P, Q;n)$ is blockwise multiplicative, written as
$\tau(P, Q;n) = \prod_{i=1}^{n/k} \tau(P[i], Q[i];k),$
where $P[i]$ and $Q[i]$ denote the Pauli operators acting on the $i$-th block. 
Throughout this paper, \textit{log} denotes the logarithm to base 2, while \textit{ln} denotes the natural logarithm. We write $[n]$ for the set $\{1,2,\ldots,n\}$, and $\|A\|_{\text{op}} = \max_{i}(\sigma_i(A))$, $\|A\|_{\mathrm{F}} = \sqrt{\sum_i \sigma_i(A)^2}$ and $\|A\|_{\mathrm{tr}} = \sum_i \sigma_i(A)$, where $\{\sigma_i(A)\}$ are the singular values of $A$.

\section{Unitary ensemble}
\begin{figure}[t]
    \centering
    \includegraphics[width=0.7\linewidth, trim=4cm 7.5cm 9cm 3.5cm, clip]{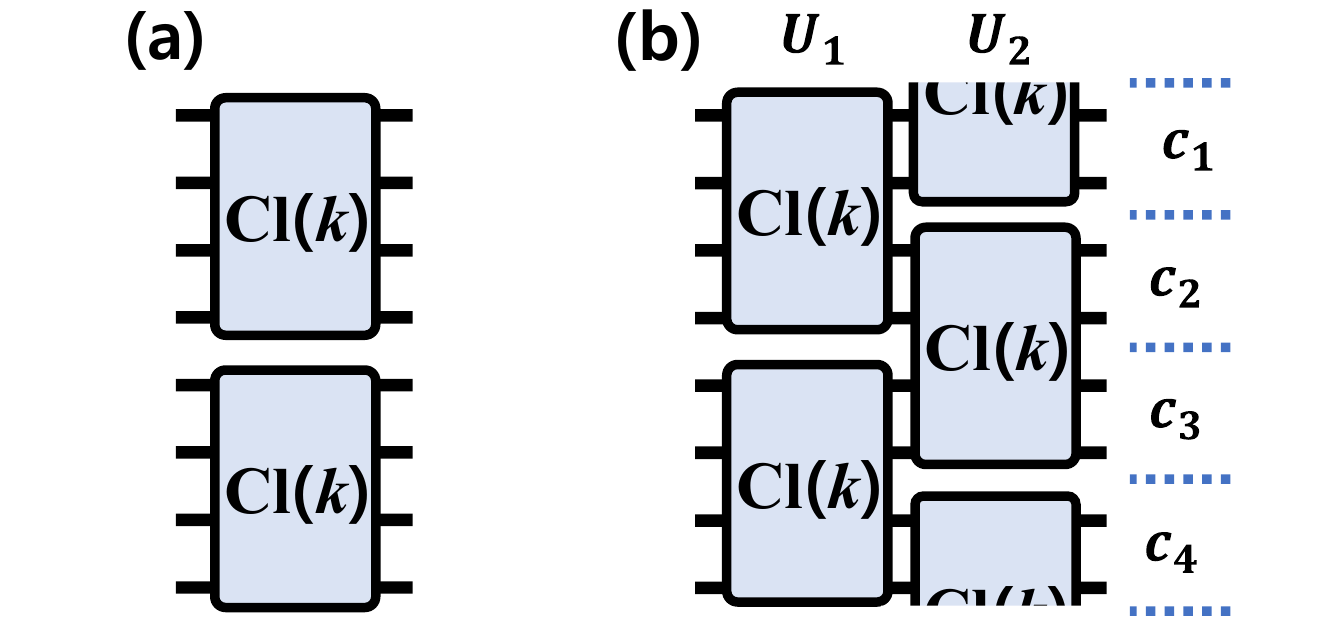}
    \caption{(a) Block random unitary ensemble, where each block is uniformly sampled from $\text{Cl}(k)$. (b) Two-layer block random unitary ensemble, with $U_1$ followed by $U_2$ on the staggered blocks, yielding $U = U_2 U_1$. The figures show the case $n=8$, $k=4$, and $m=n/k=2$. In the main text, a single block in (b) is further divided into two sub-blocks, denoted, for example, as $c_1, \ldots, c_4$.}
    \label{fig2}
\end{figure}
We consider two unitary ensembles generated by shallow-depth local circuits. 
The first is the block Clifford ensemble
$\mathcal{E}_1=\mathrm{Cl}(k)^{\otimes n/k}$ [Fig.~\ref{fig2}a], 
where each $k$-qubit block is independently sampled from the Clifford group. 
The second is a two-layer brickwork ensemble, obtained by applying two layers of $k$-qubit Clifford blocks in a staggered pattern [Fig.~\ref{fig2}b]. 
Writing the two layers as $U_1$ and $U_2$, the overall unitary is $U=U_2U_1$. 
In both cases, the circuit can be implemented by a local quantum circuit of depth $\mathcal{O}(k)$~\cite{bravyi2021clifford}. 
Thus, choosing $k=\mathcal{O}(\log n)$ gives logarithmic-depth local circuits.

Throughout this paper, we assume that $k$ divides $n$, so that there are $m=n/k$ blocks in each layer. 
We also assume that $k$ is even, allowing each block to be divided into two sub-blocks of size $k/2$.

For the unitary ensemble ${\cal E} = \mathrm{Cl}(k)^{\otimes n/k}$, $m_P = (2^k+1)^{-w_k(P)}$,
where $w_k(P)$ denotes the number of blocks in which the Pauli operator $P$ acts nontrivially. For the ensemble corresponding to Fig.~\ref{fig2}b, $m_P$ can also be explicitly calculated [See Appendix~\ref{twolayer}], but for our purpose its lower bound suffices, $m_P\geq(2^k+1)^{-n/k}$.
When using the unitary shown in Fig.~\ref{fig2}b, we introduce the notation $w_{k,1(2)}(P)$ to denote the number of blocks on which $P$ acts nontrivially in the blocks of the $U_{1(2)}$ layer. The $\tau_{1(2)}$ are defined analogously.

In both types of circuit architectures considered [Fig.~\ref{fig2}], we restrict to Clifford gates.
Then, the following two properties hold~\cite{cho2025entanglement}:
\begin{itemize}
    \item[(i)] 
    $(U, P, b)
    =
    (U, P, b)\mathbf{1}\{UPU^{\dagger}\in\pm{\cal Z}\}$.

    \item[(ii)]
    $\sum_{b}\prod_i (U,P_i,b)
    =
    2^n\delta_{\prod_i P_i,\pm I}
    \prod_i\mathbf{1}\{UP_iU^{\dagger}\in\pm{\cal Z}\}$.
\end{itemize}
Here, $(U,P,b)\coloneqq
\operatorname{Tr}(UPU^{\dagger}\ket{b}\!\bra{b})$.

\section{Sample complexity analyses}
Our main technical results consist of two parts.
First, for a rank-$r$ quantum state, we show that local random quantum circuits of depth ${\cal O}(\log n)$ achieve a sample complexity of 
${\cal O}(dr^2 \ln(d)/\epsilon^2)$, matching the known lower bound $\Omega(dr^2/\epsilon^2)$ up to a logarithmic factor~\cite{Lowe2025lower}.
Second, when no low-rank promise is available, we analyze the full-rank worst-case regime $r=d$. In this case, we prove that the optimal sample complexity $\Theta(d^3/\epsilon^2)$ can be achieved, again using local random quantum circuits of depth ${\cal O}(\log n)$. 

In the rank-$r$ case, we establish concentration of the estimator via the matrix Bernstein inequality \cite{guctua2020fast}, while in the full-rank case ($r=d$) we apply McDiarmid’s inequality \cite{acharya2025pauli1}. From these, we obtain the following two theorems.
\begin{theorem}[Rank $r$]\label{thm1:rankr}
    If the quantum state $\rho$ has rank $r$, then using the measurement circuit of Fig.~\ref{fig2}b with depth ${\cal O}(\log(n))$, 
    one can perform $(\epsilon, \delta)$-QST provided the number of samples $T = {\cal O}\left(dr^2\log(d/\delta)/\epsilon^2\right)$.
\end{theorem}
\begin{proof}[Proof sketch]
    We use the unbiased estimator $\hat{\rho}=
        \mathcal{M}^{-1}
        \!\left(
        U_1^{\dagger}U_2^{\dagger}\ket{b}\!\bra{b}U_2U_1
        \right)$.
    Since the shadow channel is diagonal in the Pauli basis, the inverse channel amplifies each Pauli component by 
    $m_P^{-1}\leq(2^k+1)^{n/k}$. 
    Applying this bound to the Pauli expansion of the estimator gives the single shot bound
    \begin{equation}
        \|\hat{\rho}\|_{\mathrm{op}}
        \leq
        (2^k+1)^{n/k}.
    \end{equation}

    It remains to control the variance parameter 
    $\|\mathbb{E}\hat{\rho}^2\|_{\mathrm{op}}$.
    Using the group structure of the $U_1$-layer Clifford ensemble 
    $\mathrm{Cl}(k)^{\otimes n/k}$, 
    the outer $U_1$ gates can be pulled outside the inverse channel:
    \[
        \hat{\rho}
        =
        U_1^{\dagger}
        \mathcal{M}^{-1}
        \!\left(
        U_2^{\dagger}\ket{b}\!\bra{b}U_2
        \right)
        U_1 .
    \]
    Therefore, $\hat{\rho}^2= U_1^{\dagger}
        \mathcal{M}^{-1}(U_2^{\dagger}\ket{b}\!\bra{b}U_2)^2
        U_1$.
    A direct second-moment calculation then gives
    \begin{equation}
        \left\|
        \mathbb{E}_{U_1,U_2,b}\hat{\rho}^2
        \right\|_{\mathrm{op}}
        \leq
        \left(1+2^{-k}\right)^{2n/k}
        \sum_{P,Q}
        \tau_2(P,Q) f_1(P,Q),
    \end{equation}
    where $f_1(P,Q)=(2^k-1)^{-w_{k,1}(PQ)}$.

    The remaining Pauli sum has a transfer-matrix structure. 
    To see this, divide each $k$-qubit block into two adjacent sub-blocks of size $k/2$.
    The factor $\tau_2(P,Q)$ couples sub-blocks according to the $U_2$ layer, while $f_1(P,Q)$ couples them according to the $U_1$ layer. 
    Grouping Pauli pairs $(P,Q)$ according to whether they are identical on each sub-block reduces the sum to
    \[
        \sum_{P,Q}\tau_2(P,Q)f_1(P,Q)
        =
        \operatorname{Tr}((FG)^{n/k}),
    \]
    where
    \begin{equation}
        F=
        \begin{pmatrix}
        1 & \frac{1}{2^k-1} \\
        \frac{1}{2^k-1} & \frac{1}{2^k-1}
        \end{pmatrix},
        \quad
        G=
        \begin{pmatrix}
        \frac{2^{2k}}{2^k+1} & \frac{2^k(2^k-1)}{2^k+1} \\
        \frac{2^k(2^k-1)}{2^k+1} & \frac{2^k(2^k-1)^2}{2^k+1}
        \end{pmatrix}
    \end{equation}
    The eigenvalues $\lambda_{\pm}$ of $FG$ satisfy, for $k\geq2$, $\lambda_{\pm}
        \leq
        2^k+\sqrt{5}\,2^{k/2}$.
    Consequently, under the condition $k2^{k/2}\ge \Omega (n)$,
    \begin{align}
        \|\mathbb{E}\hat{\rho}^2\|_{\mathrm{op}}
        &\leq 
        \left(1+2^{-k}\right)^{2n/k} (\lambda_{+}^{n/k}+\lambda_{-}^{n/k})\notag\\
        &\leq
        \mathcal{O}(2^n).
    \end{align}

    Applying the matrix Bernstein inequality to
$\hat{\rho}_T=\frac{1}{T}\sum_i\hat{\rho}_i$, and then using the rank-$r$ projection argument [Appendix~\ref{appx:thm1}], gives $T
        =
        \mathcal{O}
        \!\left(
        {dr^2\log(d/\delta)/\epsilon^2}
        \right)$.
\end{proof}

For a rank-$r$ state, the known lower bound is $\Omega(dr^2/\epsilon^2)$ in single-copy measurements \cite{Lowe2025lower, scharnhorst2025optimal}.
Therefore, our upper bound yields near-optimal sample complexity for QST of rank-$r$ states. 
Theorem~\ref{thm1:rankr} establishes that an overlap of $\log(n)$ qubits between consecutive blocks is sufficient, which corresponds to choosing $k=2\log(n)$. 

For an explicit density-matrix or Pauli transfer matrix (PTM) representation, the classical post-processing time is $\widetilde{O}(d^2r^2/\epsilon^2)$ with memory $O(d^2)$, using the naturally available measurement data format [Appendix~\ref{appx:run_time}].
Here, $\widetilde{O}$ suppresses polylogarithmic factors in $d$, $r$, and $1/\epsilon$.

The proof above uses periodic boundary conditions (PBCs) for simplicity. 
The same transfer-matrix argument extends to open boundary conditions (OBCs), changing the sample complexity by at most a constant factor. 
More generally, if the circuit decomposes into $\ell$ disconnected regions, the sample complexity acquires at most an additional factor of $2^\ell$. 
Hence, for $\ell=\mathcal O(1)$, the rank-$r$ protocol retains the same near-optimal scaling as in the PBC case, without requiring the measurement ensemble to form an $n$-qubit approximate unitary design
[see Appendix~\ref{appx:obc}].

In the worst case where the quantum state has full rank, we show that sample-optimal QST can be achieved using structurally simpler measurement circuits than those used in Theorem~\ref{thm1:rankr}.
\begin{theorem} [\textit{Full rank}]\label{thm2:fullrank}
    If the quantum state $\rho$ has rank $d$, then using the measurement circuit of Fig.~\ref{fig2}a with depth ${\cal O}(\log n)$, 
    one can perform $(\epsilon, \delta)$-QST provided the number of samples satisfies 
    $T = {\cal O}(d^3\log(1/\delta)/\epsilon^2)$.
\end{theorem}
\begin{proof}[Proof sketch]
    The unitary ensemble used here is constructed by selecting mutually unbiased bases within each block [Fig.~\ref{fig2}a], 
    so it contains $N_U=(2^k+1)^{n/k}$ unitary settings. 
    For each setting, we perform $N_S$ repetitions, so that the total number of samples is $T=N_S N_U=N_S(2^k+1)^{n/k}$.
    
    Let the experimental outcomes obtained from the $T$ measurements be denoted by 
    $\mathbf X = \{X_1, X_2, \ldots, X_T\}$. 
    For notational convenience, we denote by 
    $\{Z_1, Z_2, \ldots, Z_{N_SN_P}\}$ the subset of measurement outcomes that provide 
    actionable information for a specific Pauli operator $P$, 
    where $Z_i\in\{-1,1\}$ and $N_P = (2^k+1)^{n/k-w_k(P)}$ is the number of actionable bases. 
    Then, an unbiased estimator can be constructed as follows:
    \begin{equation}\label{main:eq:estimator_QST2}
        \hat{\rho}=\sum_P \frac{\mu_P}{2^n N_S N_P}P,
    \end{equation}
    where $\mu_P=Z_1+Z_2+\dots+Z_{N_S N_P}$ with $\mathbb{E}[Z]=\text{Tr}(\rho P)$, $\mathbb{E}[Z^2]=1$. We bound the Frobenius error $f(\rho)=f(\rho;\mathbf X)=||\rho-\hat{\rho}||_{\mathrm F}$. By using block unitary ensemble [Fig.~\ref{fig2}a], we can derive the following two key results. 
    \begin{align}
        &1.\,\mathbb{E}[f(\rho)]\leq \frac{1}{\sqrt{N_S}}\left(\frac{4^k+2^k-1}{2^k+1}\right)^{n/2k} \label{eq:thm_key_1}\\
        &2.\,f(\rho;\mathbf X)-f(\rho;\mathbf X^{(i)})\leq \frac{2(4^k+2^k-1)^{n/2k}}{N_S(2^k+1)^{n/k}} \label{eq:thm_key_2},
    \end{align}
    where $\mathbf{X}^{(i)}$ denotes the case in which only $X_i$ is different, while all others remain the same. 
    Using Eq.~\eqref{eq:thm_key_1} and Eq.~\eqref{eq:thm_key_2}, and applying McDiarmid’s inequality, the total number of samples required for QST under the assumption $k2^k= \Omega(n)$ is $T=N_S(2^k+1)^{n/k} = {\cal O}(d^3 \log(1/\delta)/\epsilon^2)$.
\end{proof}
In the full-rank case, even in the adaptive setting, the lower bound is $\Omega(d^3/\epsilon^2)$~\cite{chen2023does}, and thus our upper bound is optimal in this regime. Unlike the rank-$r$ case, for full-rank states the sample complexity saturates at the optimal scaling when the block size $k$ = $\log n$, i.e., half of that in the rank-$r$ case. For an explicit density-matrix or PTM representation, the classical post-processing time is $\widetilde{\mathcal O}(d^2)$ with memory $\mathcal O(d^2)$, using the naturally available measurement data format [Appendix~\ref{appx:run_time}]. This runtime is optimal up to polylogarithmic factors, as explicitly outputting a density matrix or PTM representation already requires $\Omega(d^2)$ time.

For circuits with depth $o(\log n)$, optimality in sample complexity is not covered by our proof; 
however, by substituting $k=1$ into our upper bound, we recover the near-optimal single-copy measurement sample complexity of ${\cal O}(10^n / \epsilon^2)$ \cite{acharya2025pauli1}, 
which is slightly larger than the known lower bound by a factor of $\sqrt{n}$ \cite{acharya2025pauli2}. 
Therefore, we expect that even in the regime depth $o(\log n)$, the circuit of Fig.~\ref{fig2}a still provides sufficiently good scaling for QST.

\section{Numerical simulations}
\begin{figure}[t]
    \centering
    \includegraphics[width=1\linewidth, trim=1cm 7.2cm 3cm 1cm, clip]{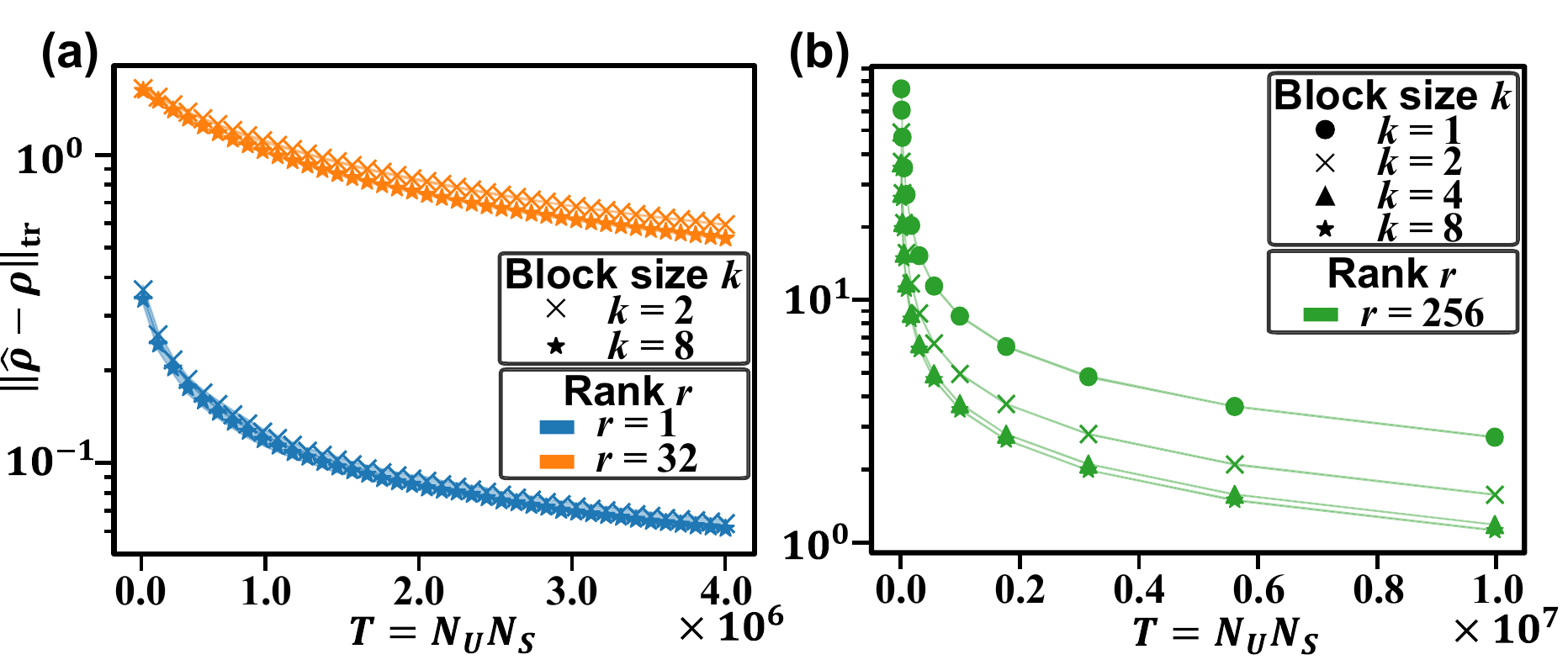}
    \caption{\label{fig:numerics}
    Numerical simulations of shallow-depth quantum state tomography (QST) on random 8-qubit states. 
    (a) Rank-$r$ QST using the two-layer block Clifford ensemble with $k=2$, compared with the global Clifford ensemble ($k=8$), for $r=1$ and $r=32$. 
    (b) Full-rank QST ($r=d=256$) using the one-layer block Clifford ensemble with $k=1,2,4,8$. 
    The trace-norm error $\|\rho-\hat{\rho}\|_{\mathrm{tr}}$ is plotted against the total number of samples $T=N_U N_S$. 
    Markers denote the mean over 10 trials, while shaded regions indicate one standard deviation.}
\end{figure}
We numerically benchmark the proposed shallow-depth QST protocols on random 8-qubit quantum states ($n = 8$). We consider the two measurement ensembles introduced in Fig.~\ref{fig2} and evaluate the reconstruction error $\|\rho-\hat{\rho}\|_{\mathrm{tr}}$ as a function of the total number of samples $T=N_U N_S$.

In Fig.~\ref{fig:numerics}a, we perform QST on random rank-$r$ quantum states using the two-layer block Clifford ensemble with block size $k=2$. As a reference, we also include results obtained from the global Clifford ensemble $\mathrm{Cl}(n)$. We perform the simulations for $r=1$ and $r=32$, sampling $N_U$ Clifford circuits and performing $N_S$ shots for each circuit, giving a total of $T=N_U N_S$ measurement outcomes. The estimator is obtained from the empirical average of Eq.~\eqref{eq:estimator_QST1}, followed by projection onto the set of rank-$r$ positive semidefinite density matrices~\cite{guctua2020fast}.

In Fig.~\ref{fig:numerics}b, we perform QST on random full-rank quantum states using the unbiased estimator in Eq.~\eqref{main:eq:estimator_QST2}. Here, we use the one-layer block Clifford ensemble with block sizes $k\in\{1,2,4,8\}$. For each $k$, we take all $N_U=(2^k+1)^{n/k}$ measurement bases obtained from mutually unbiased bases within each block, and choose $N_S$ so that the total number of samples $T=N_U N_S$ is comparable across different values of $k$. In both panels, each marker represents the average over $10$ independent trials, and shaded regions indicate the standard deviation.

\section{Discussion}
We have studied single-copy QST using measurement circuits of depth $\mathcal{O}(\log n)$. 
For rank-$r$ states, the sample complexity of our protocol matches the single-copy lower bound up to a logarithmic factor. 
For full-rank states, it achieves the optimal sample-complexity scaling. 
For explicit density-matrix or PTM output, the classical runtime improves the dependence on $r$ and $\epsilon$ over the previously known fast tomography scheme~\cite{brandao2020fast} in the rank-$r$ case, and is optimal up to polylogarithmic factors in the full-rank case. 
Although constructing an unbiased estimator for $\rho$ is straightforward, proving concentration of its empirical average is technically nontrivial. 
We address this challenge by combining the Weingarten calculus with transfer matrix techniques, obtaining a near-optimal sample complexity for rank-$r$ states. 
Moreover, for full-rank states, we show that optimal sample complexity can be achieved with structurally simpler circuits that do not form approximate unitary designs.

Our work suggests several directions for future research. First, for rank-$r$ states, it remains open whether the $\ln(d)$ factor in our upper bound can be eliminated. Second, while biased estimators suffice for certain quantities such as fidelity \cite{schuster2025random} and purity \cite{cho2025entanglement}, we showed that unbiased estimators are crucial for QST. Identifying other tasks where more refined estimators are required would be an interesting avenue for further study. Finally, in rank-$r$ case, we relied on unitary ensembles capable of forming approximate unitary designs; however, our statistical analysis did not exploit the approximate-design property. It therefore remains to be understood whether QST with small trace-norm error is possible solely under the assumption of approximate unitary designs, regardless of the circuit structure, or whether the circuit architecture itself plays an essential role.


\section*{Data Availability}
The code used to generate the data in Fig.~\ref{fig:numerics} is available at
\url{https://github.com/gyungmincho/qst-fig3-demo}.

\section*{Acknowledgements}
This work was supported by a National Research Foundation of Korea (NRF) grant funded by the Korean Government (Ministry of Science and ICT (MSIT)) (RS-2023-NR057243, RS-2023-00283291, RS-2024-00413957, SRC Center for Quantum Coherence in Condensed Matter RS-2023-00207732, RS-2023-NR077112 and Quantum Technology R\&D Leading Program (Quantum Computing) RS-2024-00442994) and a core center program grant funded by the Ministry of Education (No. 2021R1A6C101B418). 

\section*{Author contributions}
G.C. developed the theoretical framework and performed the numerical simulations. G.C. and D.K. wrote the manuscript. ChatGPT 5.5 was used to assist in discussing and refining proof ideas, checking errors, and improving the presentation of the manuscript. All proofs and results were independently written and verified by the authors, who take full responsibility for their correctness.

\bibliographystyle{apsrev4-2}  
\bibliography{main}% Produces the bibliography via BibTeX.

\onecolumngrid
% \clearpage
% \setcounter{page}{1}

% \begin{center}
%     {\large\bfseries Supplementary Material for}\\[4pt]
%     {\large\bfseries
%     ``Sample-optimal single-copy quantum state tomography
%     with shallow-depth measurements''}
% \end{center}

\appendix

\section{Related work}\label{appx:related_works}

\paragraph{Fully collective (multi-copy) measurements.}
For QST of a rank-$r$ quantum state, the smallest possible sample complexity is obtained when one is allowed to perform collective measurements on many copies of the state. 
In this setting, the lower bound on the sample complexity is saturated, and the optimal scaling for trace-distance error $\epsilon$ is
\[
    T=\Theta\!\left(\frac{dr}{\epsilon^2}\right).
\]
The lower bound was shown in Ref.~\cite{haah2016sample}, and a matching upper bound without the logarithmic factor follows from the tomography protocol of Ref.~\cite{o2016efficient}. 
These results settle the sample complexity of QST when no restriction is imposed on how many copies can be measured jointly. 
However, the corresponding measurements are highly collective and are therefore difficult to implement on near-term quantum devices.

\paragraph{$t$-copy measurements.}
A natural intermediate setting is to restrict the number of copies that can be measured jointly. 
In this model, each measurement is allowed to act on at most $t$ copies of the unknown state. 
For full-rank states, it was shown that for $t\le d^2$, the optimal sample complexity is
\begin{equation}
    T=\widetilde{\Theta}\!\left(\frac{d^3}{\sqrt{t}\epsilon^2}\right)
\end{equation}
for trace-distance error $\epsilon$~\cite{chen2024optimal}. 
This result interpolates between the single-copy and fully collective regimes, and quantifies the advantage obtained from entangling measurements across different copies. 
In particular, recovering the fully collective full-rank scaling $O(d^2/\epsilon^2)$ requires $t$ to be of order $d^2$. 
Thus, for an $n$-qubit system, a fully optimal implementation in this model requires measurements that act jointly on an exponentially large number of copies.

\paragraph{Single-copy measurements.}
In the single-copy setting, each measurement acts on only one copy of the quantum state. 
For rank-$r$ states, a lower bound of
\begin{equation}
    T=\Omega\!\left(\frac{dr^2}{\epsilon^2}\right)
\end{equation}

is known for trace-distance error $\epsilon$~\cite{haah2016sample,Lowe2025lower,scharnhorst2025optimal}. 
This scaling is achieved, up to model-dependent factors, by single-copy protocols based on structured measurements such as appropriate POVMs or random rank-one measurements~\cite{haah2016sample,kueng2017low,guctua2020fast}. 
A closely related near-optimal strategy is obtained by using projective measurements drawn from a state $2$-design, which incurs an additional logarithmic factor~\cite{Lowe2025lower}. 
These measurements are still global measurements on the $n$-qubit system, and their direct implementation generally requires circuits of depth $\mathcal{O}(n)$.

\paragraph{Single-qubit measurements for full-rank mixed states.}
A more restricted measurement model is single-qubit measurement. 
Pauli measurements are the most common example of this setting. 
For general mixed states on $n$ qubits, recent work improved the sample complexity of Pauli tomography to
\begin{equation}
    T=O\!\left(\frac{10^n}{\epsilon^2}\right)
\end{equation}
for trace-distance error $\epsilon$~\cite{acharya2025pauli1}. 
The same work also proved a lower bound of
\begin{equation}
    T=\Omega\!\left(\frac{9.118^n}{\epsilon^2}\right)
\end{equation}
for Pauli measurements, showing a separation between Pauli measurements and more general single-copy measurements. 
Subsequently, it was shown that arbitrary single-qubit measurements also require nearly
\begin{equation}
    T=\Omega\!\left(\frac{10^n}{\sqrt{n}\epsilon^2}\right)
\end{equation}
copies, up to a polynomial factor in $n$~\cite{acharya2025pauli2}. 
Therefore, single-qubit measurements are nearly optimal within their own measurement model. 
At the same time, this scaling remains separated from the $O(8^n/\epsilon^2)$scaling achievable by more general single-copy measurements.

\paragraph{Single-qubit measurements for pure states.}
For pure states, the sample complexity can be substantially smaller because of the rank-one structure. 
Recent work showed that pure-state tomography can be performed using only nonadaptive Pauli measurements with
$T=\widetilde{O}\!\left(\frac{2^n}{\delta}\right)$
copies, while outputting a state with fidelity at least $1-\delta$ with high probability~\cite{grewal2026pauli}. 
The same result also gives nearly optimal classical runtime. 
This shows that, under the promise that the unknown state is pure, single-qubit Pauli measurements are sufficient to achieve near-optimal copy complexity. 
Extending this approach to general rank-$r$ mixed states is less direct. 
One possible route is to first convert copies of a mixed state into copies of a random purification and then apply a pure-state tomography algorithm~\cite{tang2025conjugate,pelecanos2025mixed}. 
However, the random purification step itself is a genuinely multi-copy procedure. 
Therefore, this reduction does not preserve a shallow single-copy measurement structure. 
For pure states, the $r=1$ specialization of our result improves the sample complexity of Ref.~\cite{grewal2026pauli} by a $\mathrm{poly}(n)$ factor, at the cost of using logarithmic-depth circuits rather than single-qubit Pauli measurements. 
Conversely, our reconstruction produces an explicit PTM or density-matrix representation, and hence incurs an additional factor of $2^n$ in classical post-processing relative to Ref.~\cite{grewal2026pauli}. 

\section{Shadow channel}\label{appx:shadow}
In this appendix, we will explain the important properties of the shadow channel used in the main text. The shadow channel is defined by
\begin{align}\label{eqn:shadow}
    {\cal M}(\rho)
    &=\mathbb{E}_{U\sim {\cal E},b}U^{\dagger}\ket{b}\!\bra{b}U\\
    &=\mathbb{E}_{U\sim \cal E}\sum_b U^{\dagger}\ket{b}\!\bra{b}U\cdot \mathrm{Tr}(U\rho U^{\dagger} \ket{b}\!\bra{b}),
\end{align}
where ${\cal E}$ is the unitary ensemble from which $U$ is sampled uniformly and $b\in\{0,1\}^n$ is the measurement outcome on $U\rho U^{\dagger}$. By using ~\eqref{eqn:shadow} and assuming the unitary ensemble ${\cal E}$ is informationally complete, we can construct an unbiased estimator of the quantum state $\rho$ as follows:
\begin{equation}\label{eqn:unbiased}
    \hat{\rho}={\cal M}^{-1}(U^{\dagger}\ket{b}\!\bra{b}U).
\end{equation}
With \eqref{eqn:unbiased}, an unbiased estimator of the arbitrary observable $O$ is as follows:
\begin{equation}
    \hat{o}=\text{Tr}(O\hat{\rho}).
\end{equation}
The shadow channel and its inverse depend on the choice of the unitary ensemble $\mathcal E$. 
Thus, the ensemble should be chosen according to the observables or reconstruction task of interest. 
For example, if the goal is to estimate low-weight observables, then a local Clifford ensemble such as 
$\mathcal E=\mathrm{Cl}(1)^{\otimes n}$ is often appropriate. 
On the other hand, if one is interested in observables with bounded trace norm, such as 
$O=\ket{\psi}\!\bra{\psi}$, then the global Clifford ensemble $\mathcal E=\mathrm{Cl}(n)$ is a natural choice.

In the main text, we assume that our unitary ensemble consists solely of Clifford gates. Under this assumption, the following two properties hold:
\begin{align}
    &1.\quad(U, P, b)=(U, P, b)\mathbf{1}\{UPU^{\dagger}\in\pm{\cal Z}\}\label{appx:prop1}\\ 
    &2.\quad\sum_{b}\prod_i (U,P_i,b)=2^n\delta_{\prod_i P_i,\pm I}\prod_i\mathbf{1}\{UP_iU^{\dagger}\in\pm{\cal Z}\},\label{appx:prop2}
\end{align}
where $(U,P,b)\coloneqq \text{Tr}(UPU^{\dagger}\ket{b}\!\bra{b})$, ${\cal Z}=\{I,Z\}^n$ and $\mathbf{1}\{\text{True}\}=1$, otherwise 0. Eq.~\eqref{appx:prop1} is immediate since $UPU^{\dagger}$ is also a Pauli operator. It yields 0 unless it is a Pauli $Z$-string, and hence the indicator function can be freely inserted. Eq.~\eqref{appx:prop2} is checked as follows:
\begin{align}
    \sum_b\prod_i (U,P_i,b)&=\sum_b\prod_i (U,P_i,b)\mathbf{1}\{UP_iU^{\dagger}\in\pm{\cal Z}\}\\
    &=\sum_b\text{Tr}\left(\left(\prod_i UP_iU^{\dagger}\right)\ket{b}\!\bra{b}\right)\prod_i\mathbf{1}\{UP_iU^{\dagger}\in\pm{\cal Z}\}\\
    &=\sum_b\text{Tr}\left(U\left(\prod_i P_i\right)U^{\dagger}\ket{b}\!\bra{b}\right)\prod_i\mathbf{1}\{UP_iU^{\dagger}\in\pm{\cal Z}\}\\
    &=\text{Tr}\left(U\left(\prod_i P_i\right)U^{\dagger}\right)\prod_i\mathbf{1}\{UP_iU^{\dagger}\in\pm{\cal Z}\}\\
    &=2^n\delta_{\prod_i P_i, I}\prod_i\mathbf{1}\{UP_iU^{\dagger}\in\pm{\cal Z}\},
\end{align}
where $\delta$ denotes equality up to a phase. Using the above properties, we can show the Pauli operator P is an eigenoperator of the shadow channel ${\cal M}$ as follows: 
\begin{align}
    {\cal M}(P)&=\mathbb{E}_{U}\sum_{b} U^{\dagger}\ket{b}\!\bra{b}U\cdot\text{Tr}(UPU^{\dagger}\ket{b}\!\bra{b})\\
    &=\mathbb{E}_{U}\sum_{b}\sum_Q \frac{Q}{2^n}\text{Tr}(UQU^{\dagger}\ket{b}\!\bra{b})\cdot\text{Tr}(UPU^{\dagger}\ket{b}\!\bra{b})\\
    &=\mathbb{E}_{U}\sum_{b}\sum_Q \frac{Q}{2^n}\text{Tr}(UQU^{\dagger}\ket{b}\!\bra{b})\cdot\text{Tr}(UPU^{\dagger}\ket{b}\!\bra{b})\mathbf{1}\{UPU^{\dagger}\in\pm{\cal Z}\}\mathbf{1}\{UQU^{\dagger}\in\pm{\cal Z}\}\\
    &=\mathbb{E}_{U}\sum_Q \mathbf{1}\{UPU^{\dagger}\in\pm{\cal Z}\}\mathbf{1}\{UQU^{\dagger}\in\pm{\cal Z}\}\delta_{P,Q}Q\\
    &=\mathbb{E}_{U} \mathbf{1}\{UPU^{\dagger}\in\pm{\cal Z}\}P\\
    &=m_P P,
\end{align}
where $m_P\coloneqq\mathbb{E}_{U} \mathbf{1}\{UPU^{\dagger}\in\pm{\cal Z}\}$.\\

For a given unitary ensemble $\mathcal{E}$, we define the Pauli correlation function $\tau(P, Q)$ as
\begin{equation}
    \tau(P, Q; n) = \mathbb{E}_{U\sim {\cal E}}\!\left[\mathbf{1}\{UPU^{\dagger}\in\pm{\cal Z}\}\mathbf{1}\{UQU^{\dagger}\in\pm{\cal Z}\}\right],
\end{equation}
where $n$ denotes the system size on which Pauli operators $P$ and $Q$ are defined. We omit $n$ when clear from context.
For example, in the case $\mathcal{E} = \mathrm{Cl}(n)$, the $n$-qubit Clifford group, $\tau(P,Q;n)$ is given by
\begin{equation}
    \tau(P, Q;n) =
    \begin{cases}
        1 & P=I,Q=I\\
        \frac{1}{2^n+1} &P=I,Q\neq I \quad\text{or}\quad P\neq I,Q=I \quad\\
        \frac{1}{2^n+1} & P=Q\neq I\\
        \frac{2}{(2^n+1)(2^n+2)} &P\neq I, Q\neq I, P\neq Q, PQ=QP\\
        0&PQ=-QP
    \end{cases}
\end{equation}
and for $\mathcal{E} = \mathrm{Cl}(k)^{\otimes n/k}$, $\tau(P, Q;n)$ is blockwise multiplicative, written as
\begin{equation}
    \tau(P, Q;n) = \prod_{i=1}^{n/k} \tau(P[i], Q[i];k),
\end{equation}
where $P[i]$ and $Q[i]$ denote the Pauli operators acting on the $i$-th block.\\

In the main text of Theorem~\ref{thm1:rankr}, to simplify $\hat{\rho}^2$, we use the following
\begin{equation}
    V_1^{\dagger}{\cal M}(O)V_1={\cal M}(V_1^{\dagger}OV_1),
\end{equation}
where $V_1$ is the unitary sampled from $U_1$ layer ($V_1 \in \text{Cl}(k)^{\otimes n/k}$). This can be checked as follows:
\begin{align}
    V_1^{\dagger}{\cal M}(O)V_1
    &=V_1^{\dagger}\left(\mathbb{E}_{U_1, U_2}\sum_{b}U_1^{\dagger}U_2^{\dagger}\ket{b}\!\bra{b}U_2U_1\cdot\text{Tr}(U_2U_1OU_1^{\dagger}U_2^{\dagger})\right)V_1\\
    &=\mathbb{E}_{\widetilde{U}_1, U_2}\sum_{b}\widetilde{U}_1^{\dagger}U_2^{\dagger}\ket{b}\!\bra{b}U_2\widetilde{U}_1\cdot\text{Tr}(U_2\widetilde{U}_1V_1^{\dagger}OV_1\widetilde{U}_1^{\dagger}U_2^{\dagger})\\
    &={\cal M}(V_1^{\dagger}OV_1),
\end{align}
where we set $\widetilde{U}_1=U_1V_1$.
By taking $O\leftarrow{\cal M}^{-1}(O)$ and applying ${\cal M}^{-1}$ to both sides, we obtain
\begin{equation}
    {\cal M}^{-1}(V_1^{\dagger}OV_1)=V_1^{\dagger}{\cal M}^{-1}(O)V_1.
\end{equation}

\section{Limitations on biased estimator}\label{appx:biased}
Previous work \cite{schuster2025random} shows that unitary ensemble ${\cal E}$ generated from shallow quantum circuit of depth $\mathcal{O}\!\left( \log\!\left(\tfrac{n}{\epsilon}\right)\, t \,\mathrm{poly}\log(t) \right)$ form $\epsilon$-approximate unitary $t$-designs 
\begin{equation}
    (1-\epsilon)\Phi_H \preceq \Phi_{\cal E} \preceq (1+\epsilon)\Phi_H,
\end{equation}
where $\Phi \preceq \Phi'$ means that $\Phi'-\Phi$ is a completely-positive map and $\Phi_{\mathcal{E}}(\cdot)$ is defined as
\begin{equation}
    \Phi_{\mathcal{E}}(A) \coloneqq \mathbb{E}_{U \sim \mathcal{E}}
\!\left[\, U^{\otimes t} A \, U^{\dagger \otimes t} \,\right].
\end{equation}
Based on the above results, one naturally expects that circuits of depth ${\cal O}(\log n)$ 
can form approximate unitary $3$-designs, and hence that the statistical behavior of shallow-depth circuits closely approximates that of $\text{Cl}(n)$. Indeed, this has been shown in 
various examples \cite{schuster2025random}. However, in settings such as randomized measurements~\cite{elben2023randomized}, where random 
unitaries are applied followed by classical post-processing, one should be cautious in 
directly applying the above reasoning. As discussed in Appendix~\ref{appx:shadow}, obtaining an unbiased estimator requires 
first determining the shadow channel ${\cal M}$. Notably, prior work has shown that by simply using the shadow channel 
${\cal M}_{\text{Haar}}(\cdot)=\frac{1}{1+d}((\cdot)+\text{Tr}(\cdot)I)$ from Haar random unitaries ensemble, one can estimate 
the fidelity and, more generally, any observable $O$ with bounded trace norm, 
with the same sample complexity as that obtained using 
${\cal E} = \text{Cl}(n)$ \cite{schuster2025random}. A similar result also holds for purity estimation using multi-shot case \cite{cho2025entanglement}. 

One might expect that the above result remains valid for estimating other arbitrary physical 
quantities as well, but simple counterexamples can be found. 
First, in the single-shot case, when estimating purity using ${\cal M}_{\text{Haar}}$, the 
accuracy of estimation cannot be guaranteed due to the bias of the estimator, which is also 
consistent with numerical simulations \cite{cioli2025approximate}. 
Second, employing ${\cal M}_{\text{Haar}}$ for estimating a Pauli operator $P$ can lead to large errors. 
For example, consider estimating the expectation value of $O = Z_1$ using the circuit in 
Fig.~\ref{fig2}(b), where each block consists of $k$-qubit. In this case, the unbiased estimator is 
given by
\[
    \hat{o} = \text{Tr}(P\hat{\rho}) = \text{Tr}\!\left(P{\cal M}^{-1}\big(U^{\dagger}\ket{b}\!\bra{b}U\big)\right) 
    = m_P^{-1}\text{Tr}(UPU^{\dagger}\ket{b}\!\bra{b}),
\]
where $m_P = \frac{3\cdot 2^k + 1}{(2^k + 1)^3}$ provided that $k\leq n/2$ and $k$ divides $n$. 
However, if one uses ${\cal M}_{\mathrm{Haar}}$ instead, the Pauli component is rescaled as if 
$m_P = 1/(2^n+1)$, rather than by the correct value 
$m_P = (3\cdot 2^k+1)/(2^k+1)^3$. 
For $k=\mathcal O(\log n)$, the corresponding rescaling factors differ exponentially in $n$, which leads to a large bias in the estimator.

Here, we have identified the limitations of the biased estimator based on 
${\cal M}_{\text{Haar}}$, but this does not imply that all biased estimators 
are problematic. Therefore, it may be possible to employ weakly biased estimators 
that allow for efficient classical post-processing while still admitting statistical 
guarantees for accurate estimation. We leave this for future work.

\section{Proof of Theorem~\ref{thm1:rankr}}\label{appx:thm1}

In this appendix, we provide the details of the proof of Theorem~\ref{thm1:rankr}. 
Throughout the proof, we use the two-layer brickwork Clifford ensemble shown in Fig.~\ref{fig2}(b). 
For a single measurement outcome, we consider the unbiased estimator
\begin{align}
    \hat{\rho} 
    &= \mathcal{M}^{-1}\!\left(U^{\dagger}\ket{b}\!\bra{b}U\right) \\
    &= \mathcal{M}^{-1}\!\left(
    U^{\dagger}_1 U^{\dagger}_2\ket{b}\!\bra{b}U_2 U_1
    \right),
\end{align}
where $U=U_2U_1$, and $U_1$ and $U_2$ denote the two unitary layers of the measurement circuit. 
Here, $\mathcal{M}$ is the shadow channel associated with this ensemble. 
By construction,
\begin{equation}
    \mathbb{E}\hat{\rho}=\rho .
\end{equation}

Given $T$ independent measurement outcomes, we define the empirical estimator by
\begin{equation}
    \hat{\rho}_T=\frac{1}{T}\sum_{i=1}^{T}\hat{\rho}_i ,
\end{equation}
where each $\hat{\rho}_i$ is an independent copy of the single-shot estimator above. 
Thus, it remains to control the deviation of $\hat{\rho}_T$ from its mean. 
Since the target state has rank at most $r$, an operator-norm concentration
bound for the linear estimator can be converted into a trace-norm guarantee by
a randomized low-rank reconstruction step. We use the following standard
randomized block Krylov guarantee.

\begin{fact}[Randomized block Krylov low-rank approximation~\cite{musco2015randomized}]
\label{fact:block_krylov}
    Let $A\in\mathbb C^{N\times D}$, let $k\le \min\{N,D\}$, and let
    $\epsilon_{\rm alg}\in(0,1)$. 
    There is a randomized block Krylov algorithm which returns a matrix
    $Z\in\mathbb C^{N\times k}$ with orthonormal columns such that, with high
    probability,
    \begin{equation}
    \label{eq:block_krylov_guarantee}
        \left\|A-ZZ^\dagger A\right\|_{\rm op}
        \le
        (1+\epsilon_{\rm alg})
        \min_{\operatorname{rank}(B)\le k}
        \|A-B\|_{\rm op}.
    \end{equation}
    Equivalently, the right-hand side is
    $(1+\epsilon_{\rm alg})\|A-A_k\|_{\rm op}$, where
    $A_k\in\arg\min_{\operatorname{rank}(B)\le k}\|A-B\|_{\rm op}$
    is a best rank-$k$ approximation of $A$.
    
    The runtime is
    \begin{equation}
    \label{eq:block_krylov_runtime}
        O\!\left(
            \operatorname{nnz}(A)\,
            \frac{k\log D}{\sqrt{\epsilon_{\rm alg}}}
            +
            N k^2
            \frac{\log^2 D}{\epsilon_{\rm alg}}
            +
            k^3
            \frac{\log^3 D}{\epsilon_{\rm alg}^{3/2}}
        \right),
    \end{equation}
    up to lower-order terms. Here $\operatorname{nnz}(A)$ denotes the number of
    nonzero entries of $A$. Once $Z$ is obtained, the rank-$k$
    approximation $ZZ^\dagger A$ can be formed explicitly as
    $Z(Z^\dagger A)$ in $O(NDk)$ additional time. In particular, for dense
    square $A\in\mathbb C^{D\times D}$ and constant $\epsilon_{\rm alg}$,
    computing $Z$ and forming $ZZ^\dagger A$ explicitly costs
    $\widetilde O(D^2 k)$.
\end{fact}

\begin{remark}
    Ref.~\cite{musco2015randomized} states the randomized block Krylov guarantee
    for real matrices, with $ZZ^T A$ in place of $ZZ^\dagger A$. 
    We use the standard complex analogue, obtained by replacing transposes with
    adjoints and using complex Gaussian test matrices. This does not change the
    stated approximation guarantee or the asymptotic runtime.
\end{remark}

We apply Fact~\ref{fact:block_krylov} to $A=\hat\rho_T$, with $k=r$ and
$\epsilon_{\rm alg}=1$. Let
$\hat\rho_T^{\rm out}=ZZ^\dagger\hat\rho_T$ be the resulting rank-$r$
reconstruction. This output is not necessarily Hermitian or positive
semidefinite, but this is sufficient for Theorem~\ref{thm1:rankr}, which only requires a trace-norm approximation to the target state. If a physical output is desired, a rank-$r$, unit-trace PSD estimator can be obtained with the same asymptotic sample and post-processing complexity via Rayleigh--Ritz compression~\cite{halko2011finding} followed by eigenvalue projection onto the probability simplex.

Suppose that $\|\hat\rho_T-\rho\|_{\rm op}\le \eta$. Since
$\operatorname{rank}(\rho)=r$, we have
\begin{equation}
    \min_{\operatorname{rank}(B)\le r}
    \|\hat\rho_T-B\|_{\rm op}
    \le
    \|\hat\rho_T-\rho\|_{\rm op}
    \le
    \eta.
\end{equation}
Therefore, on the success event of the randomized low-rank approximation algorithm,
$\|\hat\rho_T-\hat\rho_T^{\rm out}\|_{\rm op}\le 2\eta$. By the triangle
inequality,
\begin{equation}
    \|\hat\rho_T^{\rm out}-\rho\|_{\rm op}
    \le
    \|\hat\rho_T^{\rm out}-\hat\rho_T\|_{\rm op}
    +
    \|\hat\rho_T-\rho\|_{\rm op}
    \le
    3\eta .
\end{equation}
Moreover, $\operatorname{rank}(\hat\rho_T^{\rm out}-\rho)
\le \operatorname{rank}(\hat\rho_T^{\rm out})+\operatorname{rank}(\rho)\le 2r$.
Hence,
\begin{equation}
\label{eq:block_krylov_trace_conversion}
    \|\hat\rho_T^{\rm out}-\rho\|_{\rm tr}
    \le
    2r\|\hat\rho_T^{\rm out}-\rho\|_{\rm op}
    \le
    6r\eta .
\end{equation}
Thus it suffices to prove the operator-norm concentration bound
\begin{equation}
\label{eq:sufficient_operator_bound_rankr}
    \|\hat\rho_T-\rho\|_{\rm op}
    \le
    \frac{\epsilon_{\rm tr}}{6r}.
\end{equation}

We prove the required operator-norm concentration using the matrix Bernstein inequality. 
For this purpose, it is enough to bound two quantities: the operator norm 
$\|\hat{\rho}\|_{\mathrm{op}}$, and the second moment 
$\|\mathbb{E}\hat{\rho}^2\|_{\mathrm{op}}$, which controls the variance parameter.

\begin{theorem*}[Matrix Bernstein Inequality]
    Let $X_1, X_2, \dots, X_T$ be independent random matrices in $\mathbb{C}^{d \times d}$ 
    satisfying $\mathbb{E}[X_i] = 0$ and $\|X_i\|_{\mathrm{op}} \leq R$ almost surely.
    Then, for all $t \geq 0$,
    \[
    \Pr\!\left( \left\| \sum_{i=1}^T X_i \right\|_{\mathrm{op}} \geq t \right) 
    \leq d \exp\!\left( \frac{-t^2/2}{\sigma^2 + Rt/3} \right),
    \]
    where $\sigma^2 \coloneqq 
        \left\| \sum_{i=1}^T \mathbb{E}[X_i^2] \right\|_{\mathrm{op}}$.
\end{theorem*}

We apply the inequality to
\begin{equation}
    X_i = \hat{\rho}_i-\rho .
\end{equation}
Then $\mathbb{E}X_i=0$, and
\begin{equation}
    \|X_i\|_{\mathrm{op}}
    \leq \|\hat{\rho}_i\|_{\mathrm{op}}+\|\rho\|_{\mathrm{op}}
    \leq \|\hat{\rho}_i\|_{\mathrm{op}}+1.
\end{equation}
Moreover, 
\begin{equation}
    \left\|\mathbb{E}X_i^2\right\|_{\mathrm{op}}
    \leq
    \left\|\mathbb{E}\hat{\rho}_i^2\right\|_{\mathrm{op}} .
\end{equation}
Therefore, the proof reduces to deriving upper bounds on 
$\|\hat{\rho}\|_{\mathrm{op}}$ and 
$\|\mathbb{E}\hat{\rho}^2\|_{\mathrm{op}}$. First, $\lVert \hat{\rho}\rVert_{\mathrm{op}}$ is bounded as follows: 
\begin{align}
        \left\lVert \hat{\rho} \right\rVert_{\text{op}} 
        &= \left\lVert \mathcal{M}^{-1}(U^{\dagger}_1 U^{\dagger}_2\ket{b}\!\bra{b}U_2 U_1) \right\rVert_{\text{op}} \\
        &= \left\lVert \sum_P \frac{\operatorname{Tr}(U_2U_1 P U_1^{\dagger} U_2^{\dagger}\ket{b}\!\bra{b})}{2^n} m_P^{-1} P \right\rVert_{\text{op}} \\
        &\leq \sum_P \frac{\left|\operatorname{Tr}(U_2U_1 P U_1^{\dagger}U_2^{\dagger}\ket{b}\!\bra{b})\right|}{2^n} m_P^{-1} \\
        &\leq \max_P(m_P^{-1}) \sum_P \frac{\operatorname{Tr}(U_2U_1 P U_1^{\dagger}U_2^{\dagger}\ket{b}\!\bra{b})^2}{2^n} \\
        &\leq (2^k+1)^{n/k}.
    \end{align}
It remains to bound the second-moment term, which controls the variance parameter in the matrix Bernstein inequality. 

Using the identity ${\cal M}^{-1}(V_1^\dagger O V_1) = V_1^\dagger {\cal M}^{-1}(O)V_1$, proved in Appendix~\ref{appx:shadow} for $V_1\in\mathrm{Cl}(k)^{\otimes n/k}$, we obtain
\begin{equation} \label{eq:shadow_out}
    \hat{\rho}={\cal M}^{-1}(U^{\dagger}_1 U^{\dagger}_2\ket{b}\!\bra{b}U_2 U_1)
    =U^{\dagger}_1{\cal M}^{-1}( U^{\dagger}_2\ket{b}\!\bra{b}U_2) U_1.
\end{equation}
    Then, $\hat{\rho}^2 $ is simplified to
\begin{equation}
    \hat{\rho}^2=(U^{\dagger}_1{\cal M}^{-1}( U^{\dagger}_2\ket{b}\!\bra{b}U_2) U_1)^2
    =U^{\dagger}_1{\cal M}^{-1}( U^{\dagger}_2\ket{b}\!\bra{b}U_2)^2 U_1.
\end{equation}
    Now, we can compute the average $\mathbb{E}\hat{\rho}^2$
    as follows:
    \begin{align}
    &\mathbb{E}_{U_1, U_2, b} \hat{\rho}^2
        = \mathbb{E}_{U_1, U_2} \sum_{b} 
        U^{\dagger}_1 \mathcal{M}^{-1}\!\left( U^{\dagger}_2\ket{b}\!\bra{b}U_2 \right)^2 U_1 
        \operatorname{Tr}\!\left( \rho U_1^{\dagger}U_2^{\dagger}\ket{b}\!\bra{b}U_2 U_1 \right)\\
    & = \mathbb{E}_{U_1, U_2} \sum_{b}\sum_{P,Q,R} 
        \frac{m_P^{-1} m_Q^{-1}}{8^n} \,
        \operatorname{Tr}\!\left( \rho U_1^{\dagger}RU_1 \right) 
        (U_2, P, b)(U_2, Q, b)(U_2, R, b) \, U_1^{\dagger} P Q U_1 & \\
    & = \mathbb{E}_{U_1, U_2} \sum_{P,Q} 
        \frac{m_P^{-1} m_Q^{-1}}{4^n} \,
        \mathbf{1}\!\left\{ U_2 P U_2^{\dagger} \in \pm \mathcal{Z} \right\}
        \mathbf{1}\!\left\{ U_2 Q U_2^{\dagger} \in \pm \mathcal{Z} \right\}
        \operatorname{Tr}\!\left( \rho U_1^{\dagger} P Q U_1 \right) U_1^{\dagger} P Q U_1\\
    & = \mathbb{E}_{U_1} \sum_{P,Q} 
        \frac{m_P^{-1} m_Q^{-1}}{4^n} \, \tau_{2}(P,Q) \,
        \operatorname{Tr}\!\left( \rho U_1^{\dagger} P Q U_1 \right) U_1^{\dagger} P Q U_1\\
    & = \sum_{P,Q} 
        \frac{m_P^{-1} m_Q^{-1}}{4^n} \, \tau_{2}(P,Q) \,
        \sum_{B \subseteq A} 
        \left( \frac{1}{4^k - 1} \right)^{|A|} 
        (-1)^{|A|-|B|} \, (2^k)^{|B|} \rho_B \otimes I_{A-B}, &
    \end{align}
    where $A = A(P,Q) \subseteq [n/k]$ denotes the set of block indices of the $U_1$ layer such that $P$ and $Q$ differ on that block. To obtain the last line, we use the property of unitary 2-designs. Using the above relation, we can obtain an upper bound on $\left\lVert\mathbb{E}_{U_1, U_2, b} \hat{\rho}^2\right\rVert_{\text{op}}$ as follows:
    \begin{align}
    \left\lVert \mathbb{E}_{U_1, U_2, b} \hat{\rho}^2\right\rVert_{\text{op}}
    &\leq\sum_{P,Q} 
        \frac{m_P^{-1} m_Q^{-1}}{4^n} \, \tau_{2}(P,Q)\!
        \left\lVert \sum_{B \subseteq A} 
        \left( \frac{1}{4^k - 1} \right)^{|A|}\!
        (-1)^{|A|-|B|} \, (2^k)^{|B|} \rho_B \otimes I_{A-B}\right\rVert_{\text{op}} \\
    &\leq \sum_{P,Q} \frac{m_P^{-1} m_Q^{-1}}{4^n} \, \tau_2(P,Q) 
    \left( \frac{1}{2^k-1} \right)^{w_{k,1}(PQ)} \\[6pt]
    &\leq \max_{P,Q} \left( \frac{m_P^{-1} m_Q^{-1}}{4^n} \right) 
    \sum_{P,Q} \tau_2(P,Q) 
    \left( \frac{1}{2^k-1} \right)^{w_{k,1}(PQ)} \\[6pt]
    &\leq \left( 1+\frac{1}{2^k} \right)^{2n/k} 
    \sum_{P,Q} \tau_2(P,Q) f_1(P,Q),
    \end{align}
    where $f_1(P,Q;n) =(2^k-1)^{-w_{k,1}(PQ)}=\prod_{i\in U_1 \text{layer}} f(P[i],Q[i];k)$, where $f(P,Q;k)=(2^k-1)^{-\mathbf{1}\{P\neq Q\}}$.
    We can rewrite the sum over $P$ and $Q$ into a form that is more convenient to compute.
    First, we divide each $k$-qubit block into two parts. This gives a total of $2m$ sub-blocks. Then, $\tau_2$ acts on the pairs independently \{$(1,2), (3,4)\dots, (2m-1,2m)$\}, while $f_1$ acts on the pairs independently \{$(2,3), (4,5) \dots, (2m-2,2m-1),(2m,1)$\}. We denote by $\bullet$ the case when $P[i] = Q[i]$ within a $i$-th sub-blocks, and by $\times$ otherwise. Using this notation, the given sum $\sum_{P,Q}$ can be rewritten as a sum over configurations $\textbf{c}\in\{\bullet,\times\}^{2m}$, together with a sum over pairs $(P,Q)$ compatible with each $\textbf{c}$, i.e., $\sum_{\textbf{c}}\sum_{(P,Q)  \text{ in } \textbf{c}}$. Accordingly, the above expression can be reorganized as follows:
    \begin{align}
        \left\lVert\mathbb{E}\hat{\rho}^2\right\rVert_{\text{op}}&\leq(1+1/2^k)^{2m}\sum_{P,Q}\tau_2(P,Q)f_1(P,Q)\\
        &= (1+1/2^k)^{2m}\sum_{\textbf{c}\in {\{\bullet,\times\}}^{2m}} \sum_{(P,Q)  \text{ in } \textbf{c}}\tau_2(P,Q)f_1(P,Q)\\
        &= (1+1/2^k)^{2m} \sum_{\textbf{c}\in {\{\bullet,\times\}}^{2m}}f_1(\textbf{c})\sum_{(P,Q)  \text{ in } \textbf{c}}\tau_2(P,Q)\\
        &= (1+1/2^k)^{2m} \sum_{\textbf{c}\in {\{\bullet,\times\}}^{2m}}f_1(\textbf{c})g_2(\textbf{c})\\
        &= (1+1/2^k)^{2m}\sum_{c_1, c_2,...,c_{2m}\in {\{\bullet,\times\}}}f(c_1,c_2)f(c_3,c_4)...f(c_{2m-1,2m}) \nonumber \\
        &\qquad\qquad\qquad\qquad\qquad\qquad \quad        \times g(c_2,c_3)g(c_4,c_5)...g(c_{2m},c_1)\\
        &=(1+1/2^k)^{2m}\text{Tr}((FG)^m) \label{eq:bound_rho_squred}\\
        &=(1+1/2^k)^{2m}(\lambda_{+}^{m}+\lambda_{-}^{m}),\label{eq:thm1_1}
    \end{align}
    where $g_2(\textbf{c})=\sum_{(P,Q)  \text{ in } \textbf{c}}\tau_2(P,Q)$ and $g(c_i,c_{i+1})=\sum_{(P,Q)  \text{ in } (c_i,c_{i+1})}\tau_2(P,Q;k)$. 
    Here, $F$ and $G$ are 2$\times$2 matrices defined by\\
    \begin{equation}\label{eq:F_G}
        F=
        \begin{pmatrix}
        1 & \frac{1}{2^k-1} \\
        \frac{1}{2^k-1} & \frac{1}{2^k-1}
        \end{pmatrix},
        \quad
        G=\frac{1}{2^k+1}
        \begin{pmatrix}
        2^{2k} & 2^k(2^k-1) \\
        2^k(2^k-1) & 2^k(2^k-1)^2
        \end{pmatrix}
    \end{equation}
    and $\lambda_{\pm}$ are eigenvalues of $FG$ and are given by
    \begin{align}
        \lambda_{\pm}
        &=\frac{2^k(2^{k+1}+1\pm\sqrt{2^{k+4}-7})}{2(2^k+1)}\\
        &\leq2^k+\sqrt{5\cdot2^k} \label{eq:thm1_2},
    \end{align}
    where the last inequality holds when $k\geq2$. 
    
We now compute the matrices $F$ and $G$ explicitly. For the $F$ matrix, we have $F(\bullet, \bullet)=1$, while the remaining entries 
$F(\bullet, \times)$, $F(\times, \bullet)$, and $F(\times, \times)$ are all equal to 
$1/(2^k-1)$, since the Pauli operators $P_i$ and $Q_i$ differ on each block.
For the $G$ matrix, where $G(c_1,c_2)=\sum_{(P,Q)\text{in}(c_1,c_2)}\tau(P,Q;k)$, because $G(\bullet, \times) = G(\times, \bullet)$, it suffices to compute $G(\bullet,\bullet)$, $G(\bullet,\times)$, and $G(\times,\times)$. 
\begin{figure*}[t]
    \centering
    \includegraphics[width=1\linewidth, trim=1cm 11.5cm 5.5cm 3cm, clip]{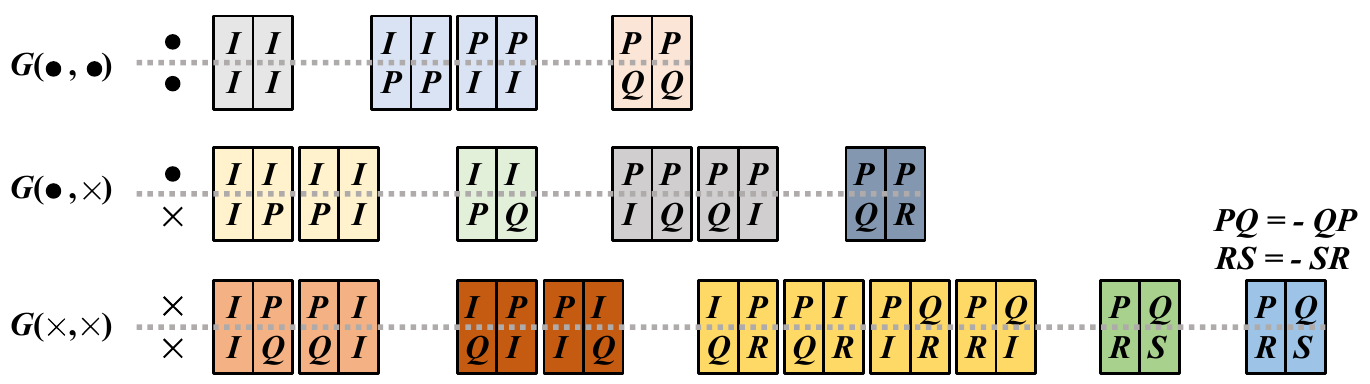}
    \caption{Case classification by matrix elements of $G$. Each vertical rectangle corresponds to a $k$-qubit Pauli operator, and two of them form a pair(e.g. $P[i]$ and $Q[i]$ in the main text). A single block is further divided into two sub-blocks by gray dashed line, and depending on whether the adjacent sub-blocks are identical or different, the cases are categorized as $\bullet$ or $\times$. $I$ denotes the identity, while $P, Q, R, S$ represent non-identity Pauli operators.}
    \label{fig4}
\end{figure*}

In Fig.~\ref{fig4}, we enumerate the Pauli-pair configurations used to compute the matrix elements of $G$. 
Each colored box represents a pair of $k$-qubit Pauli operators on one block, with the block divided into two sub-blocks. 
For each sub-block, we use $\bullet$ when the two Pauli operators are identical on that sub-block, and $\times$ otherwise. 
Within each case, we further distinguish whether the identity operator $I$ appears. 
For instance, a $\bullet$-sub-block can be of the form $(I,I)$ or $(P,P)$ with $P\neq I$, whereas a $\times$-sub-block can be of the form $(I,P)$, $(P,I)$, or $(P,Q)$ with $P\neq I$, $Q\neq I$, and $P\neq Q$. 
Only Pauli pairs that commute on the full $k$-qubit block contribute, since $\tau(P,Q;k)=0$ for anticommuting block-level pairs. 
In the $G(\times,\times)$ entry, this also includes the special case where the two Pauli operators anticommute on each sub-block, e.g., $PQ=-QP$ and $RS=-SR$, but commute on the full block. 
Configurations with the same contribution are shown in the same color in Fig.~\ref{fig4}. 
Using this classification, we obtain

\begin{align}
& G(\bullet,\bullet)= [1]+ \left[\frac{1}{2^k+1}\cdot(2^k-1)\cdot2 \right]+ \left[\frac{1}{2^k+1}\cdot(2^k-1)^2 \right] = \frac{2^{2k}}{2^k+1}, &&\\\nonumber \\
& G(\bullet,\times)= \left[\frac{1}{2^k+1}\cdot(2^k-1)\cdot2 \right]
  + \left[\frac{2}{(2^k+1)(2^k+2)}\cdot\frac{(2^k-1)(2^k-4)}{2} \right] &&\\
& \qquad\; + \left[\frac{2}{(2^k+1)(2^k+2)}\cdot(2^k-1)^2\cdot2\right]
  + \left[\frac{2}{(2^k+1)(2^k+2)}\cdot\frac{(2^k-1)^2(2^k-4)}{2} \right]= \frac{2^k(2^k-1)}{2^k+1},  &&\\\nonumber \\
& G(\times,\times)= \left[\frac{1}{2^k+1}\cdot(2^k-1)^2\cdot2 \right]
  + \left[\frac{2}{(2^k+1)(2^k+2)}\cdot(2^k-1)^2\cdot2 \right] &&\\
& + \left[\frac{2}{(2^k+1)(2^k+2)}\cdot\frac{(2^k-1)^2(2^k-4)}{2}\cdot4 \right]
  + \left[\frac{2}{(2^k+1)(2^k+2)}\cdot\frac{(2^k-1)^2(2^k-4)^2}{4} \right] &&\\
& + \left[\frac{2}{(2^k+1)(2^k+2)}\cdot\frac{(2^k-1)^22^{2k}}{4} \right]
  = \frac{2^k(2^k-1)^2}{2^k+1}. &&
\end{align}
Therefore, the matrices $F$ and $G$ are given by Eq.~\eqref{eq:F_G}, as desired.

Combining Eq.~\eqref{eq:thm1_1} and Eq.~\eqref{eq:thm1_2} under the assumption $k2^{k/2}= \Omega(n)$, we obtain the following upper bound:
    \begin{align}
        \|\mathbb{E}\hat{\rho}^2\|_{\text{op}}&\leq(1+1/2^k)^{2m}(\lambda_{+}^{m}+\lambda_{-}^{m})\\
        &\leq 2(1+1/2^k)^{2m}\lambda_{+}^{m}\\
        &\leq2(1+1/2^k)^{2m}(2^k+\sqrt{5\cdot2^k})^{m}\\
        &\leq \text{exp}(2+\sqrt{5})2^{n+1}\\
        &\leq {\cal O}(2^n)
    \end{align}
    By applying the Matrix Bernstein inequality to $\hat{\rho}_T=\frac{1}{T}\sum_{i}\hat{\rho}_i$, the following bound holds:
    \begin{equation}
        \Pr\!\left(\lVert \hat{\rho}_T - \rho \rVert_{\mathrm{op}} \geq \epsilon_{\text{op}} \right)
        \leq d \exp\!\left(-\frac{T \epsilon_{\text{op}}^2}{\lVert \mathbb{E}\hat{\rho}^2 \rVert_{\mathrm{op}} + \lVert \hat{\rho} \rVert_{\mathrm{op}} \epsilon_{\text{op}} / 3}\right).
    \end{equation}
    Using the upper bounds on $\lVert \hat{\rho} \rVert_{\mathrm{op}}$ and $\lVert \mathbb{E}\hat{\rho}^2 \rVert_{\mathrm{op}}$, and setting $\epsilon_{\mathrm{op}}=\epsilon_{\mathrm{tr}}/(6r)$, we can replace $\lVert\hat{\rho}_T - \rho \rVert_{\mathrm{op}}$ with $\lVert\hat{\rho}_T - \rho \rVert_{\mathrm{tr}}$. Consequently, the required total number of samples $T$ is
    \begin{equation}
        T = \mathcal{O}\!\left(\frac{d r^2 \ln(d/\delta)}{\epsilon^2}\right).
    \end{equation}

\section{Open boundary conditions}\label{appx:obc}
\begin{figure*}[t]
    \centering
    \includegraphics[width=0.7\linewidth, trim=0cm 9cm 5cm 2cm, clip]{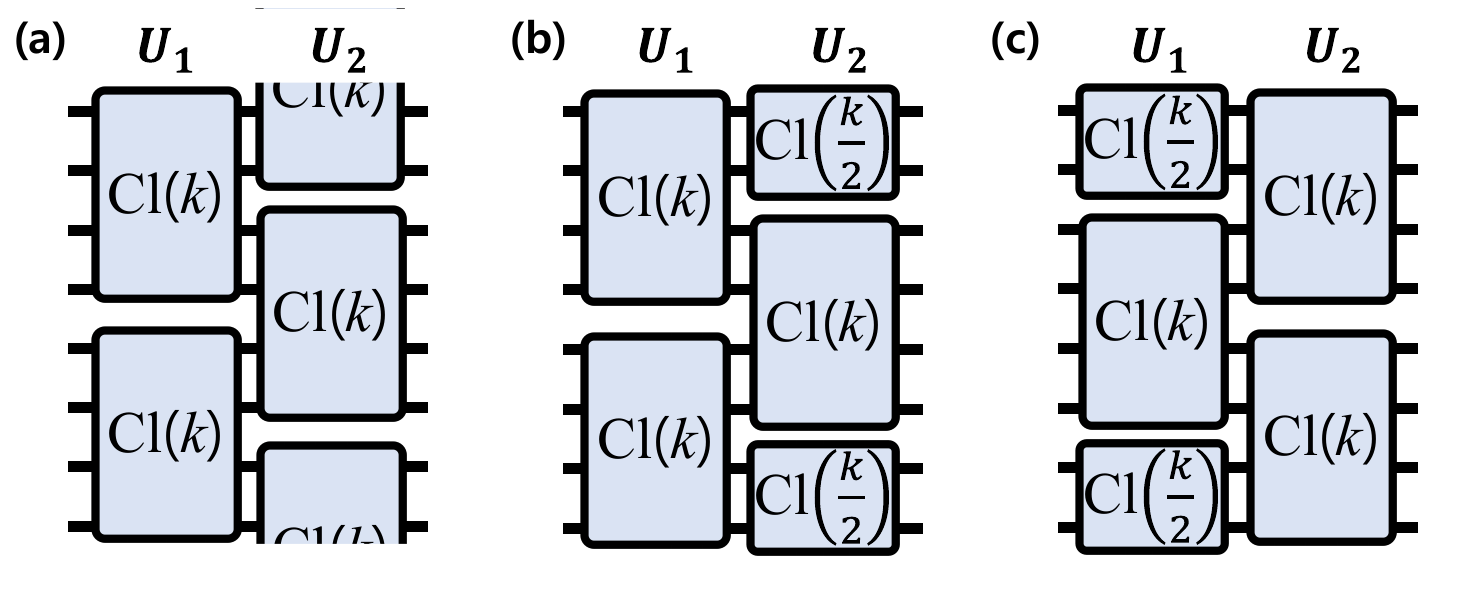}
    \caption{\label{fig5} (a) Two-layer patched brickwork circuit with periodic boundary. 
    (b) and (c) Open boundary cases where a $\text{Cl}(k)$ block is split into $\text{Cl}(k/2)^{\otimes2}$, 
    in $U_2$ and $U_1$, respectively.}
    
\end{figure*}
In Theorem~\ref{thm1:rankr}, we assumed periodic boundary conditions (PBCs) when applying the unitary [Fig.~\ref{fig5}(a)]. However, in the case of superconducting qubits, there exist connectivity constraints between qubits, which necessitate either forming a loop structure or employing additional SWAP gates. Therefore, it is currently more practical to consider open boundary conditions (OBCs). For OBCs, two cases are possible: the first is to split a $\mathrm{Cl}(k)$ into $\mathrm{Cl}(k/2)^{\otimes 2}$ in the $U_2$ layer [Fig.~\ref{fig5}(b)], and the second is to divide a $\mathrm{Cl}(k)$ in the $U_1$ layer [Fig.~\ref{fig5}(c)]. Since $\mathrm{Cl}(k/2)$ can be absorbed into the preceding or succeeding $\mathrm{Cl}(k)$, it does not need to be physically implemented in an actual experiment. Nevertheless, for conceptual clarity, we keep them.

In the PBCs case, we have computed $\text{Tr}((FG)^m)$, where
\begin{equation}
    F=
    \begin{pmatrix}
    1 & \frac{1}{2^k-1} \\
    \frac{1}{2^k-1} & \frac{1}{2^k-1}
    \end{pmatrix},
    \quad
    G=\frac{1}{2^k+1}
    \begin{pmatrix}
    2^{2k} & 2^k(2^k-1) \\
    2^k(2^k-1) & 2^k(2^k-1)^2
    \end{pmatrix},
\end{equation}
which, in the OBCs setting, is replaced by
\begin{align}
    &\text{Fig.~\ref{fig5}(b)}: \text{Tr}\left((FG)^{m-1}F\widetilde{G}\right),\\
    &\text{Fig.~\ref{fig5}(c)}: \text{Tr}\left((FG)^{m-1}\widetilde{F}G\right),
\end{align}
where 
\begin{equation}
    \widetilde{F}=
    \begin{pmatrix}
    1 & \frac{1}{2^{k/2}-1} \\
    \frac{1}{2^{k/2}-1} & \frac{1}{(2^{k/2}-1)^2}
    \end{pmatrix},
    \quad
    \widetilde{G}=
    \begin{pmatrix}
    2^{k} & 2^k(2^{k/2}-1) \\
    2^k(2^{k/2}-1) & 2^k(2^{k/2}-1)^2
    \end{pmatrix}.
\end{equation}
Let $h(X) = \text{Tr}((FG)^{m-1}X)$. We have
\begin{equation}
    h(X) = \alpha_{+}\lambda_{+}^{m-1} + \alpha_{-}\lambda_{-}^{m-1},
\end{equation}
where 
\begin{equation}\label{eq:alpha}
    \alpha_{\pm}=\frac{\text{Tr}(FGX)-\lambda_{\mp}\text{Tr}(X)}{\lambda_{\pm}-\lambda_{\mp}}
\end{equation}
and $\lambda_{\pm}$ are eigenvalues of $FG$ defined by
\begin{equation}
    \lambda_{\pm}=\frac{2^k\left(2^{k+1}+1\pm\sqrt{2^{k+4}-7}\right)}{2(2^k+1)}.
\end{equation}
For the two cases used below, $X=F\widetilde{G}$ and $X=\widetilde{F}G$, 
a direct calculation shows that $\alpha_{\pm}\geq 0$ for $k\geq2$. Then, $h(X)$ is bounded as follows:
\begin{align}
    h(X)
    &\leq\lambda_{+}^{m-1}(\alpha_{+}+\alpha_{-})\\
    &=\lambda_{+}^{m-1}\text{Tr}(X)\\
    &=\lambda_{+}^{m}\left(\frac{\text{Tr}(X)}{\lambda_{+}}\right).
\end{align}
Since $\text{Tr}(F\widetilde{G}) = 2^{k+1} = \Theta(2^k)$, $\text{Tr}(\widetilde{F}G) = \frac{2^k}{2^k+1}\left(2^{k} + 4 \cdot 2^{k/2} + 3\right) = \Theta(2^k)$, and $\lambda_{+} = \Theta(2^k)$, it follows that 
\begin{equation}
    h(F\widetilde{G}) = \mathcal O(\lambda_{+}^m),\quad h(\widetilde{F}G) = \mathcal O(\lambda_{+}^m).
\end{equation}
Thus, as in the PBC case,
\begin{equation}
    \left\lVert \mathbb{E}\hat{\rho}^2 \right\rVert_{\mathrm{op}}
    \leq
    \left(1+\frac{1}{2^k}\right)^{2m}
    \mathcal O(\lambda_+^m)
    =
    \mathcal O(2^n),
\end{equation}
under the condition $k2^{k/2}= \Omega(n)$. 
Therefore, even in the OBC setting, one can still recover near-optimal sample complexity for rank-$r$ QST.

We now generalize the preceding argument to the case in which the circuit in
Fig.~\ref{fig5}(b) consists of $\ell$ disconnected regions. The case of periodic
boundary conditions corresponds to $\ell=1$. The case of Fig.~\ref{fig5}(c) can be
treated in the same way.

From Eq.~\eqref{eq:bound_rho_squred}, under the condition
$k2^{k/2}\ge \Omega(n)$, we obtain
\begin{align}
    \left\lVert \mathbb{E}\hat{\rho}^2 \right\rVert_{\mathrm{op}}
    &\leq
    \left(1+\frac{1}{2^k}\right)^{2m}
    \operatorname{Tr}\!\left(
    (FG)^{n_1}F\widetilde{G}
    (FG)^{n_2}F\widetilde{G}
    \cdots
    (FG)^{n_{\ell}}F\widetilde{G}
    \right) \nonumber \\
    &=
    \left(1+\frac{1}{2^k}\right)^{2m}
    \prod_{j=1}^{\ell}
    \operatorname{Tr}\!\left((FG)^{n_j}F\widetilde{G}\right) \nonumber \\
    &\leq
    \left(1+\frac{1}{2^k}\right)^{2m}
    2^{\ell}
    \left(
    \frac{2^k}{2^k+\sqrt{5}\,2^{k/2}}
    \right)^{\ell}
    \left(2^k+\sqrt{5}\,2^{k/2}\right)^{m} \nonumber \\
    &\leq
    \left(1+\frac{1}{2^k}\right)^{2m}
    2^{\ell}
    \left(2^k+\sqrt{5}\,2^{k/2}\right)^{m}
    O(1) \nonumber \\
    &\leq O(2^{n+l}),
\end{align}
where $n_1+n_2+\cdots+n_{\ell}=m-\ell$. In the second line, we used the fact that
$F\widetilde{G}$ is a rank-one matrix. Indeed, if $R=F\widetilde{G}$ is written as
$R=uv^{\mathsf{T}}$, then for arbitrary compatible matrices $A_j$,
\begin{equation}
    \operatorname{Tr}\!\left(A_1 R A_2 R \cdots A_{\ell} R\right)
    =
    \prod_{j=1}^{\ell} \operatorname{Tr}(A_j R).
\end{equation}
To obtain the third line, we use the eigenvalue decomposition of $FG$. Let
$\lambda_{\pm}$ be the two eigenvalues of $FG$. Then
\begin{align}
    \operatorname{Tr}\!\left((FG)^nF\widetilde{G}\right)
    &=
    \alpha_{+}\lambda_{+}^{n}
    +
    \alpha_{-}\lambda_{-}^{n} \nonumber \\
    &\leq
    (\alpha_{+}+\alpha_{-})\lambda_{+}^{n} \nonumber \\
    &=
    \operatorname{Tr}(F\widetilde{G})\lambda_{+}^{n} \nonumber \\
    &\leq
    2\cdot 2^k
    \left(2^k+\sqrt{5}\,2^{k/2}\right)^{n},
\end{align}
where $\alpha_{\pm}$ are defined as in Eq.~\eqref{eq:alpha} with $X\leftarrow F\widetilde{G}$, and we used $\lambda_{-} \le \lambda_{+}\leq 2^k+\sqrt{5}\,2^{k/2}$.

Therefore, for $\ell=O(1)$ disconnected regions, the same sample-complexity scaling as in the PBC case is retained. 
This shows that the QST protocol considered here remains near-optimal even when the measurement ensemble decomposes into disconnected regions and therefore does not necessarily form an approximate unitary design on the full $n$-qubit system.

\section{Proof of Theorem~\ref{thm2:fullrank}}\label{appx:thm2}

In this appendix, we provide the details of the proof of Theorem~\ref{thm2:fullrank}. 
For the full-rank case, we use the block unitary ensemble shown in Fig.~\ref{fig2}(a). 
Within each $k$-qubit block, the measurement basis is chosen from a complete set of mutually unbiased bases. 
Thus, the total number of measurement bases is $N_U=(2^k+1)^{n/k}$.

For each measurement basis, we perform $N_S$ repeated measurements, so that the total number of samples is
\begin{equation}
    T = N_S(2^k+1)^{n/k}.
\end{equation}
We construct an unbiased estimator of $\rho$ by estimating each Pauli coefficient separately. 
For a Pauli operator $P$, let $N_P$ denote the number of measurement bases that provide actionable information about $P$. 
In the block unitary ensemble, this number is given by
\begin{equation}
    N_P = (2^k+1)^{n/k-w_k(P)},
\end{equation}
where $w_k(P)$ is the number of blocks on which $P$ acts nontrivially.

Let $Z_1,\ldots,Z_{N_S N_P}\in\{-1,1\}$ be the measurement outcomes associated with the actionable bases for $P$, satisfying
\begin{equation}
    \mathbb{E}[Z_i] = \operatorname{Tr}(\rho P).
\end{equation}
We define
\begin{equation}
    \mu_P = \sum_{i=1}^{N_S N_P} Z_i .
\end{equation}
Then the estimator used in the proof is
\begin{equation}\label{eq:estimator_QST2}
    \hat{\rho}
    =
    \sum_P \frac{\mu_P}{2^n N_S N_P} P .
\end{equation}
This estimator is unbiased, since
\begin{equation}
    \mathbb{E}\hat{\rho}
    =
    \sum_P \frac{\operatorname{Tr}(\rho P)}{2^n}P
    =
    \rho .
\end{equation}

It remains to prove that $\hat{\rho}$ concentrates around $\rho$ in trace norm. 
We first establish concentration in Frobenius norm and then use 
$\|A\|_{\mathrm{tr}}\leq \sqrt{d}\|A\|_{\mathrm{F}}$. 
To this end, we apply McDiarmid's inequality to
\begin{equation}
    f(\rho)=f(\rho;\mathbf{X})
    =
    \left\lVert \hat{\rho}-\rho\right\rVert_{\mathrm{F}},
\end{equation}
where $\mathbf{X}=\{X_j\}_{j=1}^{T}$ denotes the full set of measurement outcomes.

\begin{theorem*}[McDiarmid's inequality]
    Let $X_1,...,X_T$ be independent random variables, with
    \begin{equation}
        \sup_{X_1,...,X_T,X_i'}\left|f(X_1,...X_i,...X_T)-f(X_1,...X_i',...X_T) \right|\leq C_i 
    \end{equation}
    for all i=1,...,T. Then
    \begin{equation}
        \text{Pr}(|f(X_1,...,X_T)-\mathbb{E}f(X_1,...,X_T)|\geq t)\leq2\text{exp}\left(-\frac{2t^2}{\sum_{i=1}^{T}C_i^2}\right).
    \end{equation}
\end{theorem*}
To complete the proof of Theorem~\ref{thm2:fullrank}. We need to prove two key results 
\begin{align}
    &1.~\mathbb{E}[f(\rho)]
    \leq \frac{1}{\sqrt{N_S}}
    \left(\frac{4^k+2^k-1}{2^k+1}\right)^{n/2k},
    \label{eq:bound_expectation} \\
    &2.~f(\rho;\mathbf{X}) - f(\rho;\mathbf{X}^{(i)})
    \leq \frac{2(4^k+2^k-1)^{n/2k}}{N_S(2^k+1)^{n/k}}.
    \label{eq:bound_difference}
\end{align}
where $f(\rho)\!=\!f(\rho; \mathbf{X})\!=\!\left\lVert \hat{\rho}-\rho\right\rVert_{\mathrm{F}}$ and 
$\mathbf{X}\!=\!\{X_j\}_{j=1}^{T}$, with $X_j$ denoting the $j$-th measurement outcome.
\begin{proof}[proof of 1]
    \begin{align}
        \mathbb{E}f(\rho) 
        &\leq \left[\mathbb{E}f(\rho)^2\right]^{1/2}\\
        &= \left[\mathbb{E}\text{Tr}(\hat{\rho}^2) - \text{Tr}(\rho^2)\right]^{1/2}\\
        &= \left[\sum_P \frac{N_S N_P(1-\text{Tr}(\rho P)^2) + N_S^2 N_P^2  \text{Tr}(\rho P)^2}{2^n N_S^2 N_P^2} - \text{Tr}(\rho^2)\right]^{1/2}\\
        &= \left[\sum_P \frac{N_S N_P(1-\text{Tr}(\rho P)^2)+ N_S^2 N_P^2 \text{Tr}(\rho P)^2}{2^n N_S^2 N_P^2} - \text{Tr}(\rho^2)\right]^{1/2}\\
        &= \left[\sum_P \frac{1-\text{Tr}(\rho P)^2}{2^n N_S N_P} + \sum_P\frac{\text{Tr}(\rho P)^2}{2^n}- \text{Tr}(\rho^2)\right]^{1/2}\\
        &\leq \left[\sum_P \frac{1}{2^n N_S N_P}\right]^{1/2}\\
        &=\frac{1}{\sqrt{N_S}}\left(\frac{4^k+2^k-1}{2^k+1} \right)^{\frac{n}{2k}},
    \end{align}
where $N_P=(2^k+1)^{n/k-w_k(P)}$ and $\sum_P \frac{1}{N_P}$ is computed as follows:
\begin{align}
    \sum_P \frac{1}{N_P}&=\frac{1}{(2^k+1)^{n/k}}\sum_{l=0}^{n/k}\binom{n/k}{\ell}(2^k+1)^l (4^k-1)^l\\
    &=\frac{(8^k+4^k-2^k)^{n/k}}{(2^k+1)^{n/k}}\\
    &=\frac{2^n(4^k+2^k-1)^{n/k}}{(2^k+1)^{n/k}}.
\end{align}

\end{proof}
\begin{proof}[proof of 2]
    \begin{align}
        f(\rho; \mathbf{X})-f(\rho; \mathbf{X}^{(i)})
        &=\left\lVert \hat{\rho}-\rho \right\rVert_{\text{F}}-\left\lVert \hat{\rho}'-\rho \right\rVert_{\text{F}}\\
        &\leq \left\lVert \hat{\rho}-\hat{\rho}' \right\rVert_{\text{F}}\\
        &=\left[\text{Tr}((\hat{\rho}-\hat{\rho}')^2)\right]^{1/2}\\
        &=\left[ \sum_{P \triangleright U_i} \frac{(\Delta\mu_P)^2}{2^n N_S^2N_P^2} \right]^{1/2}\\
        &\leq\left[ \frac{4}{2^n N_S^2}\sum_{P \triangleright U_i} \frac{1}{N_P^2} \right]^{1/2}\\
        &=\frac{2}{N_S}\cdot\frac{(4^k+2^k-1)^{\frac{n}{2k}}}{(2^k+1)^{\frac{n}{k}}},
    \end{align}
    where $P \triangleright U_i$ denotes the Pauli operator $P$ for which the $i$-th measurement basis $U_i$ provides actionable information and $\Delta \mu_P=\mu_P-\mu_P'$. Here, $\sum_{P\triangleright U_i}\frac{1}{N_P^2}$ is computed as follows:
    \begin{align}
        \sum_{P\triangleright U_i}\frac{1}{N_P^2}
        &=\frac{1}{(2^k+1)^{2n/k}}\sum_{l=0}^{n/k}\binom{n/k}{\ell}(2^k+1)^{2l}(2^k-1)^{\ell}\\
        &=\frac{(8^k+4^k-2^k)^{n/k}}{(2^k+1)^{2n/k}}\\
        &=\frac{2^n(4^k+2^k-1)^{n/k}}{(2^k+1)^{2n/k}}.
    \end{align}
\end{proof}
Using Eq.~\eqref{eq:bound_expectation} and Eq.~\eqref{eq:bound_difference}, and applying McDiarmid’s inequality, we obtain the following:
    \begin{align}
        \text{Pr}(f(\rho)\geq\epsilon_{F})&=\text{Pr}\left(f(\rho)-\mathbb{E}f(\rho)\geq\epsilon_{F}-\mathbb{E}f(\rho)\right)\\
        &\leq \text{Pr}\left(f(\rho)-\mathbb{E}f(\rho)\geq\epsilon_{F}-\sqrt{\frac{A(n,k)}{N_S}}\right)\\
        &\leq \text{exp}\left(-\frac{2\left(\epsilon_{F}-\sqrt{\frac{A(n,k)}{N_S}}\right)^2}{N_S(2^k+1)^{n/k}\left(\frac{2(4^k+2^k-1)^{n/2k}}{N_S(2^k+1)^{n/k}}\right)^2}\right)\\
        &=\text{exp}\left(-\frac{N_S}{2A(n,k)}\left(\epsilon_{F}-\sqrt{\frac{A(n,k)}{N_S}}\right)^2\right)\\
        &=\delta,
    \end{align}
    where $A(n,k)=\left(\frac{4^k+2^k-1}{2^k+1}\right)^{n/k}$ and $N_S=\left(1+\sqrt{2\ln(1/\delta)}\right)^2A(n,k)/{\epsilon_F^2}$. In order to replace $\lVert\cdot\rVert_{\text{F}}$ with $\lVert\cdot\rVert_{\text{tr}}$ in the above result, 
    one can substitute $\epsilon_F = \epsilon_{\text{tr}}/{\sqrt{2^n}}$. 
    Therefore, the total number of samples required for QST under the assumption $k2^k= \Omega(n)$ is as follows:
    \begin{align}
        T&=N_S(2^k+1)^{n/k}\\
        &=\left(1+\sqrt{2\ln(1/\delta)}\right)^2 2^n\left(4^k+2^k-1\right)^{n/k}/\epsilon_{\text{tr}}^2\\
        &=\left(1+\sqrt{2\ln(1/\delta)}\right)^2 8^n\left(1+\frac{1}{2^k}-\frac{1}{4^k}\right)^{n/k}/\epsilon_{\text{tr}}^2\\
        &\leq \left(1+\sqrt{2\ln(1/\delta)}\right)^2 8^n\left(1+\frac{1}{2^k}\right)^{n/k}/\epsilon_{\text{tr}}^2\\
        &\leq {\cal O}(8^n/\epsilon_{\text{tr}}^2)
    \end{align}

\section{Two-layer brickwork unitary ensemble in the full rank case}\label{twolayer}
In this Appendix, we prove that using the random unitary ensemble from Fig.~\ref{fig2}(b) also achieves the optimal sample complexity of $\mathcal{O}(d^3/\epsilon^2)$ with the same circuit volume as in Fig.~\ref{fig2}(a). Although the sample complexity is identical, the classical post-processing for QST is more practical when employing the block unitary ensemble [Fig.~\ref{fig2}(a)] \cite{cho2025entanglement}.
For the proof, we need a tighter upper bound on $m_P^{-1}$ than used in the main text. To do that, we can explicitly express $m_P$ using tensor networks \cite{bertoni2024shallow, hu2025demonstration} as follows:
\begin{align}
    m_P
    &=\mathbb{E}_{U}\mathbf{1}\{UPU^{\dagger}\in\pm \mathcal{Z}\}\\
    &=\mathbb{E}_{U_1, U_2}\mathbf{1}\{U_2U_1PU^{\dagger}_1U^{\dagger}_2\in\pm \mathcal{Z}\}\\
    &=\text{Tr}(K_{(\textbf{w}_{k,1}(P))_1}K_{(\textbf{w}_{k,1}(P))_2}...K_{(\textbf{w}_{k,1}(P))_{n/k}}),
\end{align}
where $\mathbf{w}_{k,1}(P)\in\{0,1\}^{n/k}$ is a vector whose $i$-th component is $1$ if $P$ acts non-trivially on the $i$-th block of the $U_1$ layer, and $0$ otherwise. Here, $K_0$ and $K_1$ are defined as follows:
\begin{equation}
    K_0=\frac{1}{2^k+1}
    \begin{pmatrix}
    2^k+1 & 1 \\
    0 & 0
    \end{pmatrix},
    \quad
    K_1=\frac{1}{(2^k+1)^2}
    \begin{pmatrix}
    1 & 1 \\
    2^{k+1} & 2^k
    \end{pmatrix}
\end{equation}
Now, we can show that the minimum of the $m_P$ happens when $P$ acts non-trivially on all the blocks, which means $\textbf{w}_{k,1}(P)=\mathbf{1}_{n/k}$. 
\begin{proof}
We proceed by mathematical induction. Suppose that the matrix $M_{\ell}$ obtained by multiplying $K_0$ or $K_1$ a total of $\ell$ times has the form
\begin{equation}\label{induction}
    M_{\ell} =
    \begin{pmatrix}
    A_{\ell} & B_{\ell} \\
    C_{\ell} & D_{\ell}
    \end{pmatrix}, \text{with}~A_{\ell} \geq B_{\ell} \geq 0~\text{and}~C_{\ell} \geq D_{\ell} \geq 0.
\end{equation}

For $\ell=1$, the statement can be verified directly.  
Assume that it holds for some $\ell$.  
Then, for $\ell+1$, there are two possible cases:
\begin{align}
    &K_0 M_{\ell} = \frac{1}{2^k+1}
    \begin{pmatrix}
    (2^k+1)A_{\ell}+C_{\ell} & (2^k+1)B_{\ell}+D_{\ell} \\
    0 & 0
    \end{pmatrix},\\
    &K_1 M_{\ell} = \frac{1}{(2^k+1)^2}
    \begin{pmatrix}
    A_{\ell}+C_{\ell} & B_{\ell}+D_{\ell} \\
    2^k(2A_{\ell}+C_{\ell}) & 2^k(2B_{\ell}+D_{\ell})
    \end{pmatrix}.
\end{align}
If the induction hypothesis holds for $M_{\ell}$, it can be easily verified to hold for $M_{l+1}$. Therefore, any matrix obtained as a product of a sequence of $K_0$ and $K_1$ satisfies Eq.~\eqref{induction}.
Using this, we obtain 
\begin{align}
    \text{Tr}(K_0 M_{\ell}) - \text{Tr}(K_1 M_{\ell})&=\frac{1}{(2^k+1)^2}\left((2^k+1)^2A_{\ell}+(2^k+1)C_{\ell}-A_{\ell}-C_{\ell}-2^k(2B_{\ell}+D_{\ell})\right)\\
    &=\frac{1}{(2^k+1)^2}\left(2^{2k}A_{\ell}+2^{k+1}(A_{\ell}-B_{\ell})+2^k(C_{\ell}-D_{\ell}) \right)\geq 0,
\end{align}
and hence the presence of $K_0$ in $M_{\ell}$ always gives a larger value compared to the case without it. Therefore, the minimum of $m_P$ occurs when $\textbf{w}_{k,1}(P)=(1,1,\cdots,1)$.
\end{proof}
Using the above result, a minimum of $m_P$ is given by
\begin{align}
    \min_P m_P &= \text{Tr}(K_1^{n/k})\\
    &=\mu_{+}^{n/k}+\mu_{-}^{n/k},
\end{align}
where $\mu_{\pm}$ are eigenvalues of $K_1$ and are given by
\begin{equation}
    \mu_{\pm}=\frac{z+1\pm \sqrt{z^2+6z+1}}{2(z+1)^2},
\end{equation}
where $z=2^k$. Then, the maximum of $m_P^{-1}$ is as follows: 
\begin{align}
    \max_P m_P^{-1} &= \frac{1}{\mu_{+}^{n/k}+\mu_{-}^{n/k}}\\
    &\leq\frac{1}{|\mu_{+}|^{n/k}-|\mu_{-}|^{n/k}}\\
    &=\frac{1}{|\mu_{+}|^{n/k}(1-|\mu_{-}/\mu_{+}|^{n/k})}\\
    &\leq {\cal O}(1/|\mu_{+}|^{n/k})
\end{align}
The last line holds for sufficiently large $k$, as $\left|\frac{\mu_{-}}{\mu_{+}}\right|=\Theta(2^{-k})$. And $1/|\mu_{+}|$ is bound on
\begin{align}
    1/|\mu_{+}|&\leq z+\frac{3}{z+1}\\
    &=2^k+\frac{3}{2^k+1}.
\end{align}
Therefore, when $k=\tfrac{1}{2}\log(n)$, we obtain
\begin{align}
    \max_P m_P^{-1} &\leq \left(2^k+\frac{3}{2^k+1}\right)^{n/k} \\
    &\leq 2^n\left(1+\frac{3}{2^{2k}}\right)^{n/k} \\
    &\leq \mathcal{O}(2^n).
\end{align}
Based on this result, we can now prove the concentration inequality.  

Let an unbiased estimator $\hat{\rho}_T$ be defined as
\begin{align}
    \hat{\rho}_T &= \frac{1}{N_U}\sum_{i=1}^{N_U}\hat{\rho}_{U_i}, \\
    \hat{\rho}_{U} &= \frac{1}{N_S}\sum_{j=1}^{N_S}\mathcal{M}^{-1}(U^{\dagger}\ket{b_j}\bra{b_j}U),
\end{align}
where $N_U$ denotes the number of sampled unitaries, $N_S$ denotes the number of repeated measurements for a given unitary $U$, and $T=N_UN_S$.  

As before, the Frobenius norm
\[
    f(\rho)=f(\rho;\mathbf{X})=\|\hat{\rho}_T-\rho\|_{\text{F}}
\]
can be bounded using McDiarmid's inequality. Let $\ell\!\coloneqq\!\max_P m_P^{-1}$. Then, the following two properties hold:
\begin{align}
    &1.\ \mathbb{E}[f(\rho)] \leq \sqrt{\frac{2^n L}{N_UN_S}+\frac{\ell}{N_U}}, \\
    &2.\ f(\rho;\mathbf{X})-f(\rho;\mathbf{X}^{(i)}) \leq \frac{2L}{N_UN_S}.
\end{align}
Applying these two properties to McDiarmid's inequality, we obtain
\begin{equation}
    \Pr\!\left(f(\rho)\geq \epsilon_{\text{F}}\right) \leq \delta
\end{equation}
whenever
\begin{equation}
    N_U N_S = \frac{\ell}{\epsilon_{\text{F}}^2}\left(\sqrt{2^n+N_S}+L^{1/2}\ln(2/\delta)\right)^2.
\end{equation}
Let us assume $N_S=\Theta(2^n)$, then 
\begin{align}
    N_UN_S &=\mathcal{O}\left(\frac{\ell}{\epsilon_{\text{F}}^2}\text{max}(2^n,L)\right)\\
    &=\mathcal{O}\left(\frac{L^2}{\epsilon_{\text{F}}^2}\right).
\end{align}
Therefore, by taking $\epsilon_{\text{F}}=\epsilon_{\text{tr}}/\sqrt{2^n}$, we can ensure that 
\begin{equation}
    \text{Pr}(\|\hat{\rho}_T-\rho\|_{\text{tr}}\geq \epsilon_{\text{tr}}) \leq \delta,
\end{equation}
when $T=N_UN_S=\mathcal{O}(\frac{d^3}{\epsilon^2})$, $N_S=\mathcal{O}(2^n)$ and $k=\frac{1}{2}\log n$. Note that $k=\frac{1}{2}\log n$ in the two-layer brickwork unitary ensemble has the same circuit volume as $k=\log n$ in the block unitary ensemble.

\section{Numerical simulations}\label{appx:numerical_simulations}

In this appendix, we describe the numerical simulations presented in Fig.~\ref{fig:numerics}. 
All simulations were performed for $n=8$ qubits, so that $d=2^n=256$. 
Random target states were generated using the \texttt{random\_density\_matrix} function in \texttt{qiskit.quantum\_info}, and random Clifford circuits were generated using \texttt{random\_clifford}~\cite{qiskit2024}. 
The simulations include statistical shot noise from finite measurement samples, but do not include gate errors or measurement errors.

For the rank-$r$ simulations in Fig.~\ref{fig:numerics}(a), we considered $r=1$ and $r=32$. 
The measurement unitaries were sampled either from the global Clifford ensemble $\mathrm{Cl}(n)$ or from the two-layer block Clifford ensemble shown in Fig.~\ref{fig2}(b). 
For the two-layer ensemble, we used block size $k=2$, and each $k$-qubit block was sampled independently from $\mathrm{Cl}(k)$. 
For each sampled measurement basis $U_i$, we performed $N_S$ repeated measurements. 
Although the proof of Theorem~\ref{thm1:rankr} is stated for the single-shot setting, i.e., $N_S=1$, grouping several shots under the same basis is convenient for numerical post-processing. 
The estimator used in the simulation was
\begin{align}
    \hat{\rho}_T
    =
    \frac{1}{N_U}\frac{1}{N_S}\sum_{i=1}^{N_U}\sum_{j=1}^{N_S}
    \mathcal{M}^{-1}\!\left(
    U_i^{\dagger}\ket{b_{ij}}\!\bra{b_{ij}}U_i
    \right),
\end{align}
where $b_{ij}$ is the $j$-th measurement outcome obtained from the basis $U_i$, and $T=N_U N_S$ is the total number of samples. 
In the simulations, we fixed $N_U=4000$ and varied $N_S$ to obtain different values of $T$. 
After forming the empirical estimator, we projected it onto the set of rank-$r$ positive semidefinite density matrices, following the standard post-processing used for low-rank QST~\cite{guctua2020fast}. 
Concretely, we diagonalized the Hermitian part of the estimator, kept the leading $r$ eigenvectors, projected the corresponding eigenvalues onto the probability simplex, and reconstructed a trace-one positive semidefinite matrix.

For the full-rank simulations in Fig.~\ref{fig:numerics}(b), we used the one-layer block Clifford ensemble shown in Fig.~\ref{fig2}(a). 
We considered block sizes $k\in\{1,2,4,8\}$. 
For each $k$, the measurement bases were chosen as tensor products of a complete set of mutually unbiased bases on each $k$-qubit block, giving
$N_U = (2^k+1)^{n/k}$ measurement bases in total. 
The estimator was the Pauli-coefficient estimator in Eq.~\eqref{eq:estimator_QST2}. 
For each $k$, the number of repeated measurements $N_S$ was chosen so that the total sample number $T=N_U N_S$ was comparable across different block sizes.

For both panels, the reconstruction error was measured by the trace norm 
$\|\rho-\hat{\rho}\|_{\mathrm{tr}}$. 
Each data point represents the mean over $10$ independent trials, and the shaded region indicates one standard deviation.

\section{Runtime analyses}\label{appx:run_time}

In this appendix, we analyze the classical post-processing time required to construct an explicit representation of the estimator used in the two QST protocols. 
The runtime depends on the data type in which the measurement results are supplied. 
Throughout this section, we assume the standard experimental data format in which, for each measurement circuit $U_i$, the $N_S$ repeated measurement outcomes are stored as a histogram
$h_i:\{0,1\}^n\to\mathbb{N}$, where $h_i(x)$ is the number of times the bitstring $x$ is observed. 
Equivalently, the input is a list of pairs $(U_i,h_i)$ for $i=1,\ldots,N_U$. 
This is the natural output format of present cloud-based quantum-computing platforms: for example, \texttt{Qiskit} stores circuit-execution results as a counts dictionary whose keys are measured bitstrings and whose values are the corresponding shot counts~\cite{qiskit2024,qiskit_counts_docs}. 
If one instead starts from a raw list of all $T=N_UN_S$ shots, then an additional input-reading and histogram-construction cost of order at least $\Omega(T)$ is unavoidable. 
The runtime estimates below refer to the post-processing cost after this standard histogram aggregation has been performed.

We represent operators in the Pauli transfer matrix (PTM) representation, i.e., as vectors of Pauli coefficients. 
Let $\mathcal{P}_n=\{I,X,Y,Z\}^n$ denote the $n$-qubit Pauli basis. 
For an $n$-qubit operator $A$, we define its PTM vector by
\begin{equation}
    |A\rangle\!\rangle
    =
    \sum_{P\in\mathcal{P}_n}
    \alpha_P(A)|P\rangle\!\rangle,
    \qquad
    \alpha_P(A)=\operatorname{Tr}(PA),
    \label{eq:runtime_ptm_vector}
\end{equation}
where $|P\rangle\!\rangle$ denotes the normalized basis vector associated with the Pauli operator $P$. 
Equivalently, the operator $A$ is recovered from its PTM vector as
$A=2^{-n}\sum_{P\in\mathcal{P}_n}\alpha_P(A)P$.
Thus a PTM vector $|\rho\rangle\!\rangle$ has length $d^2$. 
When an explicit density-matrix representation is required, the conversion from the PTM representation to the computational-basis matrix representation can be performed using standard fast Walsh--Hadamard-transform-type routines in $\widetilde{O}(d^2)$ time and $O(d^2)$ memory~\cite{fino1976unified}. 
This cost is unavoidable up to logarithmic factors if one explicitly outputs a dense $d\times d$ matrix.

For each fixed Clifford measurement circuit $U_i$, computational-basis measurements provide estimates of only $d$ Pauli coefficients.
Let $Z^a=\bigotimes_{\ell=1}^n Z_\ell^{a_\ell}$ for $a\in\{0,1\}^n$.
Since $U_i$ is Clifford, for each $a$ there exist $P_{i,a}\in\{I,X,Y,Z\}^n$ and $s_{i,a}\in\{\pm1\}$ such that
\begin{equation}
    U_iP_{i,a}U_i^\dagger=s_{i,a}Z^a.
    \label{eq:runtime_actionable_pauli}
\end{equation}
Thus, $\{P_{i,a}\}_{a\in\{0,1\}^n}$ are precisely the $d$ Pauli observables whose expectation values can be estimated from measurements in the basis specified by $U_i$. Then the empirical estimate of the corresponding expectation value is
\begin{equation}
    \widehat{z}_i(a)
    =
    \sum_{x\in\{0,1\}^n}
    \frac{h_i(x)}{N_S}(-1)^{a\cdot x},
    \qquad
    \widehat{\operatorname{Tr}(\rho P_{i,a})}
    =
    s_{i,a}\widehat{z}_i(a).
    \label{eq:appx_fwht}
\end{equation}
All values $\widehat{z}_i(a)$ can be computed from the histogram $h_i$ using one fast Walsh--Hadamard transform (FWHT) in $\widetilde{O}(d)$ time. 
Since the circuits are Clifford, the map $Z^a\mapsto U_i^\dagger Z^aU_i$ can also be enumerated for all $a$ in $\widetilde{O}(d)$ time using a stabilizer tableau representation. 
Therefore, each measurement basis contributes to the PTM vector in $\widetilde{O}(d)$ time, independently of $N_S$ once the outcomes have been stored as a histogram.

We first consider the rank-$r$ protocol of Theorem~\ref{thm1:rankr}. 
For multi-shot setting, the estimator can be written as
\begin{equation}
    \hat{\rho}_T
    =
    \frac{1}{T}
    \sum_{i=1}^{N_U}
    \sum_{j=1}^{N_S}
    \mathcal{M}^{-1}
    \left(
    U_i^\dagger |b_{ij}\rangle\!\langle b_{ij}|U_i
    \right),
    \qquad
    T=N_UN_S.
    \label{eq:runtime_rankr_estimator}
\end{equation}
The proof of Theorem~\ref{thm1:rankr} is stated for the single-shot setting $N_S=1$, but Eq.~\eqref{eq:runtime_rankr_estimator} is the natural multi-shot implementation of the same unbiased estimator. 
Because the measurement ensemble consists of Clifford circuits, the shadow channel is diagonal in the Pauli basis, $\mathcal{M}(P)=m_PP$, and hence applying $\mathcal{M}^{-1}$ amounts to multiplying each Pauli coefficient by $m_P^{-1}$. 
The coefficients $m_P$ for the shallow brickwork ensemble can be computed by the transfer-matrix method described in Appendix~\ref{twolayer}. 
For an explicit PTM output, precomputing or storing the required diagonal factors costs at most $\widetilde{O}(d^2)$ time and $O(d^2)$ memory, which is not larger than the cost of explicitly storing the final estimator.

Using Eq.~\eqref{eq:appx_fwht}, for each measurement basis $U_i$ we first apply the FWHT to the histogram $h_i$ to obtain all $\widehat{z}_i(a)$ in $\widetilde{O}(d)$ time. 
We then multiply the corresponding Clifford signs $s_{i,a}$ as a post-processing step and update the PTM coefficient associated with $P_i^a$. 
Combining this per-basis FWHT update with the diagonal action of $\mathcal{M}^{-1}$, the post-processing time for constructing the PTM representation of the estimator used in the rank-$r$ case is
\begin{equation}
    \mathrm{Time}_{r}
    =
    \widetilde{O}(N_Ud+d^2),
    \qquad
    \mathrm{Memory}_{r}=O(d^2).
    \label{eq:runtime_rankr_general}
\end{equation}
In the single-shot implementation used in Theorem~\ref{thm1:rankr}, we have $N_U=T$ and $N_S=1$. 
Using the sample complexity $T=\widetilde{O}(dr^2/\epsilon^2)$, Eq.~\eqref{eq:runtime_rankr_general} gives
\begin{equation}
    \mathrm{Time}_{r}
    =
    \widetilde{O}\!\left(\frac{d^2r^2}{\epsilon^2}\right),
    \qquad
    \mathrm{Memory}_{r}=O(d^2).
    \label{eq:runtime_rankr_final}
\end{equation}
For the final rank-$r$ reconstruction, we use the randomized low-rank
approximation procedure in Fact~\ref{fact:block_krylov}. Applied to the dense
estimator $A=\hat{\rho}_T\in\mathbb C^{d\times d}$ with target rank $k=r$
and constant algorithmic accuracy $\epsilon_{\rm alg}$, this step computes $Z$ and forms
\[
    \hat{\rho}_T^{\rm out}=ZZ^\dagger \hat{\rho}_T
\]
in $\widetilde O(d^2r)$ time. This avoids the $O(d^3)$ worst-case cost of obtaining the projected
least-squares (PLS) estimator~\cite{guctua2020fast}, which would require a full
eigendecomposition.
Hence the final reconstruction cost is dominated by
Eq.~\eqref{eq:runtime_rankr_final}.

We note that the runtime bound in Eq.~\eqref{eq:runtime_rankr_final} already improves on the asymptotic scaling of the previously known fast QST algorithm~\cite{brandao2020fast}, even before exploiting the repeated-shot structure. A more refined analysis in the multi-shot setting can further reduce the post-processing cost, since the reconstruction can be performed from per-basis histograms rather than from individual measurement outcomes.

We next consider the full-rank protocol of Theorem~\ref{thm2:fullrank}. 
In this case, the estimator is built directly from Pauli coefficients and does not require applying the inverse shadow channel. 
For each Pauli operator $P$, let $N_P$ be the number of measurement bases that provide actionable information about $P$, and let $\mu_P$ be the sum of the corresponding $\pm1$ outcomes over all repeated shots. 
The estimator is
\begin{equation}
    \hat{\rho}
    =
    \sum_P
    \frac{\mu_P}{2^n N_SN_P}P.
    \label{eq:runtime_fullrank_estimator}
\end{equation}
For the block measurement ensemble of Fig.~\ref{fig2}(a), the number of bases is $N_U=(2^k+1)^{n/k}$. 
As in the rank-$r$ case, one FWHT per basis computes all $d$ Pauli estimates in $\widetilde{O}(d)$ time. 
The different normalization factors $N_P^{-1}$ depend only on the block support $w_k(P)$ and can be applied while updating the PTM vector.
Therefore, for histogram input, the post-processing time for constructing the estimator used in the full-rank case in the PTM representation is
\begin{equation}
    \mathrm{Time}_{\mathrm{full}}
    =
    \widetilde{O}(N_Ud),
    \qquad
    \mathrm{Memory}_{\mathrm{full}}=O(d^2).
    \label{eq:runtime_fullrank_general}
\end{equation}
Choosing $k=\log n$ gives $N_U=(2^k+1)^{n/k}=\mathcal{O}(d)$. 
Thus Eq.~\eqref{eq:runtime_fullrank_general} becomes
\begin{equation}
    \mathrm{Time}_{\mathrm{full}}
    =
    \widetilde{O}(d^2),
    \qquad
    \mathrm{Memory}_{\mathrm{full}}=O(d^2).
    \label{eq:runtime_fullrank_final}
\end{equation}
This is nearly optimal for an explicit PTM or density-matrix output, since merely writing down $d^2$ coefficients already requires $\Omega(d^2)$ time and memory. 
The separation between the sample complexity and the post-processing time comes from the repeated-shot structure of the full-rank protocol. 
Although the protocol uses $T=\mathcal{O}(d^3/\epsilon^2)$ samples in total, these samples are grouped into $N_U=\mathcal{O}(d)$ distinct measurement bases, with many repetitions per basis. 
After the outcomes for each basis have been aggregated into a histogram $h_i$, the reconstruction uses $h_i$ rather than the individual shot outcomes.

\end{document}